\documentclass[11pt]{article}
\usepackage[margin=1in]{geometry}
\usepackage{mathtools, amssymb}
    \def\OmZ{\Omega_Z}
    \def\OO{\mathcal O}
    \def\OZ{\mathcal O_Z}
    \def\leq{\leqslant}
    \def\geq{\geqslant}
    \def\FF{\mathbb F}
    \def\AA{\mathcal A}
    \def\CC{\mathcal C}
    \def\DD{\mathcal D}
    \def\xx{\mathsf x}
    \def\yy{\mathsf y}
    \def\xxx{\mathbf x}
    \def\yyy{\mathbf y}
    \DeclareMathOperator\poly{poly}
    \def\bma#1{\begin{bmatrix}#1\end{bmatrix}}

\usepackage{amsthm}
    \newtheorem{theorem}{Theorem}
    \newtheorem{lemma}[theorem]{Lemma}
    \newtheorem{remark}[theorem]{Remark}
    \newtheorem{challenge}[theorem]{Challenge}
    \newtheorem{proposition}[theorem]{Proposition}

\usepackage[rgb, dvipsnames, svgnames]{xcolor}
    \definecolorseries{noyellow}{hsb}{last}
        [hsb]{0.3,1,0.8}[hsb]{1.1,1,0.8}

\usepackage{tikz-cd, booktabs}
    \usetikzlibrary{calc, matrix, shapes,
        perspective, decorations.pathreplacing}
    \tikzset{every picture/.style={line cap=round, line join=round}}

\usepackage{snaptodo}
    \snaptodoset{margin block/.style={font=\tiny\ttfamily}}
    \def\simple#1{\PackageWarning{thinbold}{delete todo}}
    \def\venkat#1{\PackageWarning{thinbold}{delete todo}}

\usepackage[colorlinks, allcolors=blue!50!magenta!50!black]{hyperref}

\begin{document}

                                 \title
           {Nonadaptive Noise-Resilient Group Testing with \\
          Order-Optimal Tests and Fast-and-Reliable Decoding}
                                    
            \author{Venkatesan Guruswami \and Hsin-Po Wang}
                                    
                         \def\day#1\year{\year}
                                    
                               \maketitle
                                    
\footnotetext[0\def\thefootnote{$\bigstar$}]{
    Research supported in part by
    NSF grant CCF-2210823 and a Simons Investigator Award.
    Emails: \{venkatg, simple\} @berkeley.edu.
}

\begin{abstract}
    Group testing (GT) is the Boolean version of spare signal recovery
    and, due to its simplicity, a marketplace for ideas that can be
    brought to bear upon related problems, such as heavy hitters,
    compressed sensing, and multiple access channels.  The definition of
    a ``good'' GT varies from one buyer to another, but it generally
    includes
    (i) usage of nonadaptive tests,
    (ii) limiting to $O(k \log n)$ tests,
    (iii) resiliency to test noise,
    (iv) $O(k \poly(\log n))$ decoding time, and
    (v) lack of mistakes.
    In this paper, we propose \emph{Gacha GT}.  Gacha is an elementary
    and self-contained, versatile and unified scheme that, for the first
    time, satisfies all criteria for a fairly large region of
    parameters, namely when $\log k < \log(n)^{1-1/O(1)}$.  Outside this
    parameter region, Gacha can be specialized to outperform the
    state-of-the-art partial-recovery GTs, exact-recovery GTs, and
    worst-case GTs.
    
\smallskip

    The new idea Gacha brings to the market is a redesigned
    Reed--Solomon code for probabilistic list-decoding at diminishing
    code rates over reasonably-large alphabets.  Normally, list-decoding
    a vanilla Reed--Solomon code is equivalent to the nontrivial task of
    identifying the subsets of points that fit low-degree polynomials.
    In this paper, we explicitly tell the decoder which points belong to
    the same polynomial, thus reducing the complexity and enabling the
    improvement on GT.
\end{abstract}

\section{Introduction}    

    Group testing (GT) is a powerful strategy initially developed for
    syphilis detection during World War II \cite{Dor43}, later expanding
    its utility to track diseases like AIDS \cite{GaH89, TLP95, TTB15}
    and, more recently, to manage the Covid-19 pandemic \cite{AlE22,
    MNB21, GAR21, SLW20}.  GT has been extensively documented
    \cite{DuH99, DPW00, AJS19} and taught \cite{VAM23}, illustrating
    its impact.  Beyond its medical origins, GT has found diverse
    applications.  For example, in wireless communication, GT aids in
    identifying active devices \cite{BMT84, Wol85} as it is analogous
    to the goal of identifying infected individuals in the clinical
    setting, which has applications to mMTC and IoT \cite{RoE21}.
    Other applications of GT include heavy hitters \cite{CoM05} and
    compressed sensing \cite{GIS08}.  See Table~\ref{tab:models} on
    page~\pageref{tab:models} for the counterparts of ``individual''
    and ``infection'' and ``tests'' in other applications.

    Stimulated by applications, many variants of GT had emerged.
    While binary GT uses boolean disjunction to generate test results,
    quantitative GT \cite{GHK22} uses matrix--vector multiplication
    instead, which is related to coin-weighing with a pointer scale.
    Quantitative GT should not be confused with compressed sensing
    as the latter allows non-binary (complex number!) matrices.
    Semi-quantitative GT \cite{CGM21} is a relaxation of quantitative GT
    in that the outputs are intervals instead of integers.  Tropical GT
    \cite{TropicalGT23} is a variant of semi-quantitative GT that uses
    tropical arithmetic.  Threshold GT \cite{BKC19} is a further
    relaxation that use a threshold $\theta$ to determine if the output
    is positive or negative.  Generalized GT \cite{CJZ23} is more
    flexible than threshold GT in that the probability of positive
    outcome is a monotonic function in the result of the matrix--vector
    product.  Finally, one-bit compressed sensing \cite{MaM24} is
    compressed sensing followed by erasing all digits while keeping
    signs, and coin-weighing with a balance scale is one-bit compressed
    sensing with ternary ($0$ and $\pm1$) matrices.  See also
    Table~\ref{tab:quant} on page~\pageref{tab:quant} for a summary.

    GT is also deeply connected to the notion of list-disjunct matrices
    and list-recoverable codes.  The work that best justifies these
    connections is that by Cheraghchi and Nakos \cite{ChN20}.  This work
    first constructs list-disjunct matrices \cite[Theorem~11]{ChN20} and
    then stacks them on top of each other to obtain corollaries in GT,
    heavy hitters, and compressed sensing.  The GT corollary in
    \cite{ChN20} is concurrent to \cite{PrS20}, and their constants are
    improved in \cite{BonsaiGT}.  Other works that use list-disjunct
    matrices include \cite{INR10, NPR11, IKW19}.  Meanwhile, GT works
    \cite{INR10, NPR11}, compressed sensing works \cite{NPR12, GNP13,
    GLP14}, a multiple access work \cite{ARF22}, and a heavy hitter work
    \cite{DoW22} all employ list-recoverable codes to facilitate the
    recovery of sparse signals.  Another work that illustrates the
    adaptability of the GT toolchain is SAFFRON \cite{LCP19}, which is
    an expander-graph--based GT strategy that is extended to an
    asynchronous multiple access strategy called A-SAFFRON.
    See \ref{app:imply} for more concrete connections.

    Among the applications and variants, a common scene is that $m$, the
    number of measurements, is roughly $k \log n$, the number of active
    signals times the number of bits needed to described a signal.  This
    coincides with the information-theoretical lower bound in the binary
    GT case.  These works underpin our observation that binary GT serves
    a fundamental role in the family of sparse signal recovery by
    being the most primitive (instead of complex numbers, there are
    only zeros and ones) and the most disruptive (unlike addition,
    boolean disjunction is not invertible).  Given how restrictive
    these constraints are, techniques effective in GT are very likely
    extendable to other scenarios.  That and the fact that we are a
    constant scalar away from the information-theoretical lower bounds
    are what motivate us to study binary GT.

    While the cross-application potential motivates the research of GT,
    there was not a one-size-for-all GT scheme and so the specific
    criteria of interest vary across scenarios.  In computer forensics,
    for instance, a hacker will modify some files once and only once,
    forcing the defender to prepare \textbf{nonadaptive} checksums.
    When it comes to tracking contagious diseases, \textbf{minimizing}
    the total number of expensive and invasive \textbf{tests} is
    paramount.  For wireless communication, on the other hand, the
    primary challenge lies in managing the prevalent, inherent
    \textbf{noise} within signals.  As for finding heavy hitters in a
    data stream, a \textbf{fast} decoder must avoid enumerating over all
    possible items to ensure efficiency.  Lastly, in all applications,
    the common goal is to \textbf{minimize} the number of
    \textbf{mistakes}.  See \ref{app:results} for tables of
    the state of the art.

    In this paper, we aim to construct a GT scheme that single-handedly
    satisfies said set of diverse criteria.  These five criteria, in
    particular, cannot be met by using list-disjunct matrices because
    list-disjunct matrices are paired with a decoder that maintains a
    list of possibly infected people, but with false negative tests this
    list no longer contains all infected people.  GROTESQUE \cite{CJB17}
    and SAFFRON \cite{LCP19} also have difficulties achieving
    order-optimal parameters because GROTESQUE cannot peel off
    doubletons and SAFFRON cannot peel off tripletons\footnote{ In
    \cite{LCP19}, a singleton is a time--frequency slot that is used by
    one active device.  A doubleton is a time--frequency slot used by
    two, which is too noisy for GROTESQUE but SAFFRON pays a constant
    scalar to be able to decode them.  That said, a tripleton, a slot
    used by three, is too noisy for both to decode.}, and so every
    tripleton marks a loss of information.  Neither is bit-mixing coding
    \cite{BCS21} applicable in our situation because its decoder
    enumerates over a large set of masks; more precisely, the size of
    this set, $s$, determines the mistake probability as $k^2/s$,
    meaning that it trades decoding complexity for mistake probability
    inverse-linearly, while we demand a better exchange rate.  See
    \ref{app:tricks} for a more thorough review of techniques.

    Our approach to meet all five criteria is to combine techniques from
    a list of previous works.  We base the new GT scheme off the
    framework of \emph{writing on dirty paper via expander graphs} from
    GROTESQUE \cite{CJB17} and SAFFRON \cite{LCP19}.  We incorporate the
    idea from Ngo, Porat, and Rudra's works \cite{NPR11, NPR12} of using
    list-recoverable codes as a reduction lemma to evolve parameters
    $(m, k, n)$ recursively.  We decode the list-recoverable codes by
    first clustering the data fragments similar to what \cite{LNN19}
    does; our codes assume a new data structure, which resembles hash
    tables, different from that of \cite{LNN19}, which resembles
    adjacency lists.  Lastly, we borrow Barg and Zémor's
    linear-complexity capacity-achieving codes with positive error
    exponents\footnote{ It does not have to be Barg and Zémor's codes,
    but we need all three criteria---linear complexity,
    capacity-achieving, and positive error exponent---and this is one of
    the few off-the-shelf options.} \cite{BaZ02}, prove that they
    possess good Hamming properties using Shannon-type arguments, and
    apply them to counter noisy test results.  Combining all, we propose
    \emph{Gacha GT}.

    \begin{theorem} [Main theorem]                      \label{thm:main}
        Let $Z$ be any binary-input channel that models the test noise
        and $\OZ$ hide a constant that depends only on $Z$.  Let $\sigma
        \geq 1$ and $\tau \geq 2$ be free integer parameters.  Gacha is
        a randomized GT scheme that uses $m = \OZ(\sigma k \log(n)
        2^{\OO(\tau)})$ nonadaptive tests and decoding complexity
        $\OZ(\sigma k \poly(\log(\sigma n)) 2^{\OO(\tau)})$ to find $k$
        sick people in a population of $n$.  Averaged over the
        randomness of the set of sick people, the configuration matrix,
        and the channel $Z$, Gacha will produce $k \exp(-\sigma
        \log_2(n)^{1-1/\tau})$ or fewer false negatives and false
        positives.
    \end{theorem}

    See Figure~\ref{fig:pyramid} on page~\pageref{fig:pyramid},
    Figure~\ref{fig:evolve} on page~\pageref{fig:evolve}, and
    Figure~\ref{fig:noisy} on page~\pageref{fig:noisy}
    for a sketch of proof.

\vspace{-2ex}
\paragraph{We organize our paper as follows.}
    Section~\ref{sec:prelim} states the problem and some implications of
    the main theorem.
    Section~\ref{sec:toy} presents the key of Gacha: a special
    list-decodable code.
    Section~\ref{sec:gadget} generalizes the list decoding idea to a
    concatenable list recovery idea.
    Section~\ref{sec:noise} explains how to handle noisy test results.
    Section~\ref{sec:main} wraps up the proof of the main theorem.

    In \ref{app:tricks} we catalogs some GT techniques to serve a bigger
    picture to readers.
    In \ref{app:results} we classify existing GT results into three big
    tables so readers can compare them to ours.
    In \ref{app:imply} we discuss how GT results would influence related
    problems.

\section{Problem statement and the marketplace of solutions}
                                                      \label{sec:prelim}

\subsection{Notations and problem statement}

    Let $n$ and $k$ be integers; $n > k > 0$.  We will use the disease
    control narrative throughout the paper so $n$ is the population and
    $k$ is the number of sick people.  The Mother Nature chooses the set
    of sick people uniformly at random---each with probability $1/\binom
    nk$.  Let column vector
    \[
        x =
        \bma{
            x_1
            \\[-5pt] \vdots
            \\[-4pt] x_n
        }
        \in \{0, 1\}^{n\times1}
    \]
    contain indicators where $x_j = 1$ means that the $j$th person is
    sick, $x_j = 0$ if healthy, for $j \in [n]$.

    Let $A \in \{0, 1\}^{m\times n}$ be a matrix; $m$ is the number of
    tests.  We call $A$ a \emph{configuration matrix}.  Define column
    vector
    \[
        y =
        \bma{
            y_1
            \\[-3pt] \vdots
            \\[-4pt] y_m
        }
        \in \{0, 1\}^{m\times1}
        \qquad
        \text{by}
        \qquad
        y_i \coloneqq \bigvee_{j=1}^n (A_{ij} \wedge x_j).
    \]
    They are called the test results.  An equivalent definition of $y$
    is $y \coloneqq \min(Ax, 1)$, where $Ax$ is the usual matrix--vector
    multiplication

    Denote by $Z$ some binary-input noisy channel.  It can be
    represented by a triple $Z = (\Sigma, \mu_0, \mu_1)$ where $\Sigma$
    is the output alphabet, $\mu_0$ is a distribution on $\Sigma$ that
    controls the output when the input is $0$, and $\mu_1$ is a
    distribution on $\Sigma$ that controls the output when the input is
    $1$.  We say that tests are noisy if each test result $y_i$, $i \in
    [m]$, is post-processed by an iid copy of $Z$ and only the channel
    output $z_i \coloneqq Z(y_i) \in \Sigma$ is observable.  See
    Figure~\ref{fig:channel} on page~\pageref{fig:channel} for some
    examples of noisy channels.

\subsection{The decoding problem}                      \label{sec:unify}

    The problem of noiseless GT is to solve for $x$ given $A$ and $y$.
    Denote by $\hat x = \DD(A, y) \in \{0, 1\}^n$ the estimate of $x$,
    where $\DD$ denotes the decoding algorithm.  The problem of noisy
    GT is to solve for $x$ given $A$, $Z$, and $z$.  For the latter
    case, we denote the estimate of $x$ by $\hat x = \DD(A, Z, z) \in
    \{0, 1\}^n$.

    For any index $j \in [n]$, $(x_j, \hat x_j) = (0, 1)$ is called a
    \emph{false positive} (FP).  If $(x_j, \hat x_j) = (1, 0)$, it is
    called a \emph{false negative} (FN).  In literature, there are about
    three levels of decoding quality.
    \begin{itemize}
        \itemsep=0ex
        \item Requiring that $\hat x = x$ always holds when tests are
            noiseless.  This is called \emph{combinatorial} GT,
            \emph{zero error}, or the \emph{for-all} setting.
        \begin{itemize}
            \vspace{-1ex}
            \item[$\circ$]
                When tests are noisy, we speak of the \emph{worst-case}
                failure probability where $x$ is adversarial and $Z$ is
                independently random; we want that $\hat x \neq x$ with
                low probability over the randomness from $A$, $Z$ and
                $\DD$.
        \end{itemize}
        \item Requiring that $\hat x = x$ with probability $1 - o(1)$
            over the randomness from $x$, $A$, $Z$, and $\DD$.  This is
            called \emph{probabilistic} GT, \emph{small error}, or the
            \emph{for-each} setting.
        \item To some extend $0 < \varepsilon < 1$, allowing
            $\varepsilon k$ FNs and FPs on average (averaging over the
            randomness from $x$, $A$, $Z$, and $\DD$).
            This is called \emph{partial recovery},
            as opposed to \emph{exact recovery}.
        \begin{itemize}
            \vspace{-1ex}
            \item[$\circ$] A close notion is
                \emph{uniform approximation} \cite{Che13}, meaning that
                $|\hat x - x|$ is always less than $\varepsilon k$.
        \end{itemize}
    \end{itemize}
    Note that there is not a qualitative difference between these
    variants, but rather it is a matter of controlling the number of FNs
    and FPs.
    \begin{itemize}
        \itemsep=0ex
        \item Partial recovery allows $\OO(k)$ FPs and FNs on average.
        \item Exact recovery is a consequences of having $o(1)$ FNs and
            FPs on average by Markov's inequality.
        \item Worst-case GT and zero-error GT are consequence of $< 1 /
            \binom nk$ FNs and FPs.  That is, if we prove that the
            expected number of FPs and FNs is $< 1/2\binom nk$, then the
            worst-case failure probability over all $x$ is $< 1/2$,
            which rounds down to $0$ when tests are noiseless and $A$
            and $\DD$ are de-randomized.
    \end{itemize}
    
    The shift of viewpoint from bounding the probability of mistake to
    counting FNs and FPs unifies GT regimes and is parallel to the
    $\ell_p$-error in the context of heavy hitters and compressed
    sensing.  We will see how to derive corollaries of the main theorem
    using this viewpoint.

\subsection{Our new results and implications}

    Our new results stem from a construction of probabilistically
    list-decodable codes that work over reasonably-large alphabets at
    vanishing code rates, against the list-decoding tradition that the
    alphabet size depends square-exponentially on the gap to capacity
    \cite[Fig.~1]{GuX22}.  This construction, in the second half this
    paper, will be generalized to probabilistically list-recoverable
    codes that work over reasonably-large alphabets at vanishing code
    rates, which is also outside the parameter regime that is usually
    considered \cite[Fig.~1]{HeW18} \cite[Tables 1 and 2]{KRR21}.  We
    therefore advertise our work as motivating this ad hoc parameter
    regime as well as providing a proof-of-concept solution.

    Our new code yields an \emph{\textbf{elementary and self-contained,
    versatile and unified}} improvements in GT: Our new scheme, has an
    extra $\exp(-\log_2(n)^{1-1/\tau})$-term in the number of FPs and
    FNs.  Such an improvement is independent of parameter regime.
    Hence, using the viewpoint at the end of Section~\ref{sec:unify},
    Gacha pushes the current bounds in partial-recovery GT,
    exact-recovery GT, and worst-case GT compared to the quite distinct
    earlier approaches for different settings, as can be seen in
    Table~\ref{tab:compare}.

    \begin{itemize}
        \itemsep=0ex
        \item \textbf{Partial-recovery Gacha} outperforms
            partial-recovery GROTESQUE by producing fewer FNs and FPs.
        \item \textbf{Exact-recovery Gacha} outperforms exact-recovery
            GROTESQUE by using fewer tests.
        \item \textbf{Order-optimal Gacha} outperforms
            Price--Scarlett--Tan with a lower complexity in the
            parameter region $\log k < \log(n)^{1-1/\OO(1)}$.
        \item For when $x$ is adversarial, there exists a de-randomized
            \textbf{worst-case Gacha} using fewer tests than
            Atia--Saligrama's construction.  Moreover, Gacha can work
            with all channels.
    \end{itemize}

\begin{table*}
    \caption{
        $(\nu, \kappa) \coloneqq (\log_2 n, \log_2 k)$.
        All tests are noisy; $\OZ$ is a constant depending on $Z$.
        See also Tables \ref{tab:k^2} to \ref{tab:both}.
    }                                                \label{tab:compare}
    \bigskip
    \leftskip-1in plus1in
    \rightskip-1in plus1in
    \begin{tabular}{cccc}
        \toprule

        Name / reference / specialization
        & $m$ (\#tests) & $\DD$'s complexity & Remark
        
        \\ \midrule

        Partial-recovery GROTESQUE \cite[Cor~8]{CJB17}
        & $\OZ(k \nu)$ & $\OZ(k \nu)$ & $\varepsilon k$ FPs \& FNs

        \\ \textbf{Partial-Recovery Gacha}
        $(\tau, \sigma) \leftarrow (2, 1)$
        & $\OZ(k \nu)$ & $\OZ(k \poly(\nu))$
        & $ k \exp(-\nu^{1-1/\tau})$ FPs \& FNs

        \\ \cmidrule(l{3em}r{3em}){1-4}

        Exact-recovery GROTESQUE \cite[Thm~2]{CJB17}
        & $\OZ(k \nu \kappa)$ & $\OZ(k(\nu + \kappa^2))$
        & $o(1)$ FPs \& FNs 

        \\ \textbf{Exact-recovery Gacha}
        $(\tau, \sigma) \leftarrow (2, 1 + \kappa / \sqrt\nu)$
        & $\OZ(k (\nu + \kappa \sqrt\nu))$ & $\OZ(k \poly(\nu))$
        & $o(1)$ FNs \& FPs 

        \\ \cmidrule(l{3em}r{3em}){1-4}

        Price--Scarlett--Tan \cite[Thm~4.1]{PST23}
        & $\OZ(t k \nu)$ & $\OZ((k (\nu - \kappa))^{1+\varepsilon})$
        & $k^{1-t\varepsilon}$ FPs \& FNs

        \\ \textbf{Order-optimal Gacha}
        $(\tau, \sigma) \leftarrow (\OO(1), 1)$
        & $\OZ(k \nu)$ & $\OZ(k \poly(\nu))$
        & $o(1)$ FNs \& FPs if $\kappa < \nu^{1-1/\OO(1)}$

        \\ \cmidrule(l{3em}r{3em}){1-4}

        Atia--Saligrama \cite[Rem~VI.2]{AtS12}
        & $\OZ(k^2 \nu \kappa^2)$ & unspecified
        & FN channel; worst-case

        \\ \textbf{Worst-case Gacha}
        $(\tau, \sigma) \leftarrow (\sqrt{\log_2 \nu}, k2^\tau)$
        & $\OZ\bigl( k^2 \nu 2^{\OO(\sqrt{\log\nu})} \bigr)$
        & $\OZ(k^2 \poly(\nu))$
        & any channel; worst-case

        \\ \bottomrule
    \end{tabular}
\end{table*}

\section{A Toy Example of Gacha GT: Assuming
         \texorpdfstring{$\log_2 n \geq \log_2(k)^2$}{}} \label{sec:toy}

    Recall that there are $k$ sick people among $n$ people Let $\nu
    \coloneqq \log_2 n$.  We let every person associate to a unique
    $\nu$-bit string called \emph{phone number}.  Throughout this
    section, we assume that $\sqrt\nu$ is an a multiple of $6$ and
    $\log_2 k \leq \sqrt\nu$ in order to discuss a toy example of Gacha
    GT.  We begin with a premature blueprint that uses $k \nu$ tests but
    does not quite work.

\subsection{A premature blueprint that does not work
            (but is inspirational)}

    To begin, reshape each phone number to a $\sqrt\nu$-by-$\sqrt\nu$
    array.\footnote{Or an $\alpha$-by-$\beta$ array as long as
    $\alpha \beta = \nu$.}\footnote{ If one chooses $(\alpha, \beta)
    = (1, \nu)$, he will recover GROTESQUE and SAFFRON
    (Figure~\ref{fig:horizontal}).  In the end, the failure probability
    is $e^{-\Theta(\alpha)} + e^{-\Theta(\beta)}$; this explains
    the choice of $\alpha = \beta = \sqrt\nu$ and the suboptimality
    of GROTESQUE and SAFFRON.}\footnote{ If one chooses
    $(\alpha, \beta) = (\nu, 1)$, he will recover bit-mixing coding
    (Figure~\ref{fig:vertical}).  In that work, the decoder enumerates
    over all $k^2 \poly(\log k)$ masks, leading to a suboptimal decoding
    complexity we do not want to pay.} Divide $k \nu$ tests into $k
    \sqrt\nu$ batches; each batch contains $\sqrt\nu$ tests.  We let
    each person choose $\sqrt\nu$ random batches.  For the $s$th batch
    she chooses, $s \in [\sqrt\nu]$, she will participate in the tests
    at places where the $s$th row of her phone number has $1$, as is
    illustrated in Figure~\ref{fig:QR}.  More precisely, suppose that
    $b_1 < \dotsb < b_{\sqrt\nu} \in [k\sqrt\nu]$ are the random batches
    she chooses, then she will participate in the $(b_s\sqrt\nu + t)$th
    test if the $((s - 1) \sqrt\nu + t)$th digit of her phone number is
    $1$, for all $s, t \in [\sqrt\nu]$.

    When $k = 1$, we can recover the phone number of the only sick
    people by reading off the test results.  But for general $k$, the
    blueprint does not work due to a sequence of difficulties.
    
    Firstly, we cannot prevent sick people from choosing the same batch
    of tests.  When this happens, this batch becomes the bitwise-or of
    the corresponding rows of the phone numbers and gives us corrupted
    information.  To fix that, we can double the number of test in each
    batch and ask every person to participate in the $(2(b_s\sqrt\nu +
    t) - d)$th test, where $d$ is the the $((s - 1) \sqrt\nu + t)$th
    digit of her phone number, for all $s, t \in [\nu]$.  In other
    words, we apply the Thue--Morse rules $0 \leadsto 01$ and $1
    \leadsto 10$ to turn each row of a phone number into a
    length-$2\sqrt\nu$ weight-$\sqrt\nu$ vector.  Obviously, the
    bitwise-or of two different constant-weight vectors has a higher
    weight, which is the giveaway that a batch of tests contains two or
    more sick people.

    Being able to tell if a batch is a bitwise-or, the second difficulty
    we face is that some information of the phone numbers are
    permanently erased, so with high probability we cannot recover the
    phone numbers of all the sick people.  To fix that, we can multiply
    the number of batches by some constant $1/r$ and protect the columns
    of the phone number by a $\sqrt\nu$-dimensional rate-$r$
    error-correcting code.  By choosing $r$ properly, one can show that
    with high probability we lose no information.  Thus, in principle,
    we can recover the phone numbers.

    Now that we know the information is there, it remains to show that
    we indeed are able to decode the error-correcting codes and extract
    the information.  But that is exactly where the difficulty lies:
    There is not an easy way to tell which batches come from the first
    sick person and which batches come from the second sick person, and
    so on.  To this end, we isolate the difficulty and pose it as a
    standalone coding-theoretic challenge in a standalone subsection
    below.

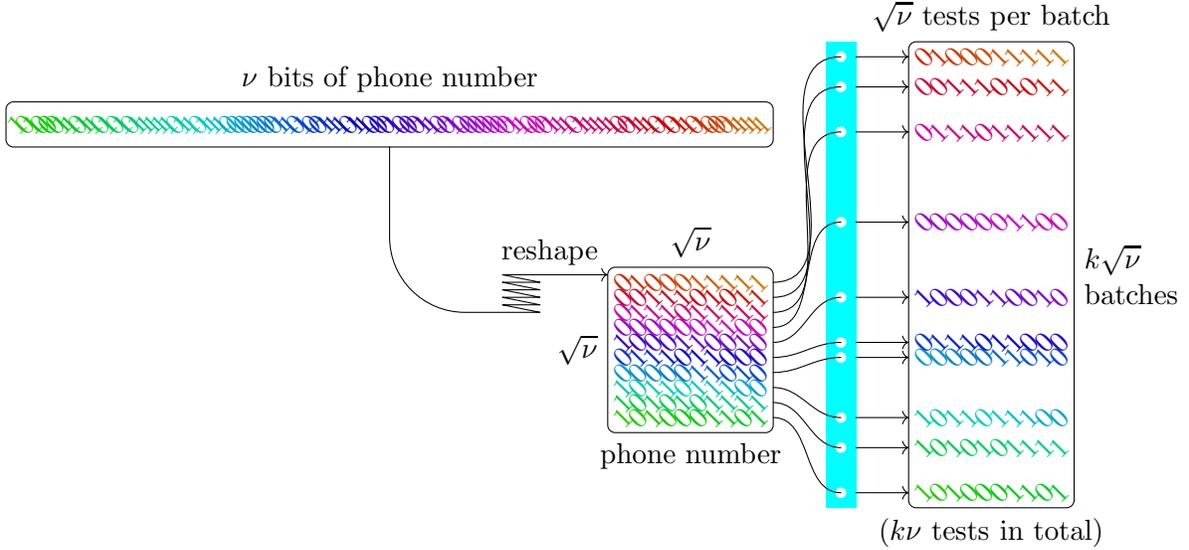
\begin{figure*}
    \centering
    \begin{tikzpicture}
        \resetcolorseries[99]{noyellow}
        \begin{scope} [shift={(-6, 2)}]
            \pgfmathsetseed{8881616}
            \foreach \x in {0, ..., 99}{
                \pgfmathtruncatemacro\b{random(0, 1)}
                \draw [{noyellow!![\x]}]
                    (\x/10-99/20, 0) node [rotate=45] {$\b$};
            }
            \draw [rounded corners=0.1cm]
                (-5.1, -0.3) rectangle (5.1, 0.3)
                (0, 0.3) node [above] {$\nu$ bits of phone number}
            ;
            \draw [->]
                (0, -0.3) -- ++(0, -1.2) arc (180:270:1) -- ++(1, 0)
                -- ++(-0.5, 0.1) -- ++(0.5, 0)
                -- ++(-0.5, 0.1) -- ++(0.5, 0)
                -- ++(-0.5, 0.1) -- ++(0.5, 0)
                -- ++(-0.5, 0.1) -- ++(0.5, 0)
                -- ++(-0.5, 0.1) -- ++(0.5, 0)
                -- +(0.9, 0) node [above left] {reshape}
            ;
        \end{scope}
        \begin{scope} [shift={(-2, -1)}]
            \pgfmathsetseed{8881616}
            \def\c{0}
            \foreach \y in {0, ..., 9}{
                \foreach \x in {0, ..., 9}{
                    \pgfmathtruncatemacro\b{random(0, 1)}
                    \draw [{noyellow!![\c]}]
                        (\x/5-9/10, \y/5-9/10) node [rotate=45] {$\b$};
                    \xdef\c{\the\numexpr\c+1}
                }
            }
            \draw [rounded corners=0.1cm]
                (-1.1, -1.1) rectangle (1.1, 1.1)
                (0, 1.1) node [above] {$\sqrt\nu$}
                (-1.1, 0) node [left] {$\sqrt\nu$}
                (0, -1.1) node [below] {phone number}
                (1.1, 0) coordinate (left)
            ;
        \end{scope}
        \begin{scope}
            \fill [cyan]
                (-0.2, 3.1) rectangle (0.2, -3.1)
                (0, 0) coordinate (middle)
            ;
            \foreach \Y in {0,3,5,9,10,13,18,24,27,29} {
                \fill [white] (0, \Y/5-29/10) circle (2pt);
            }
        \end{scope}
        \begin{scope} [shift={(2, 0)}]
            \pgfmathsetseed{8881616}
            \def\c{0}
            \foreach [count=\y] \Y in {0,3,5,9,10,13,18,24,27,29} {
                \foreach \x in {0, ..., 9}{
                    \pgfmathtruncatemacro\b{random(0, 1)}
                    \draw [{noyellow!![\c]}]
                        (\x/5-9/10, \Y/5-29/10) node [rotate=45] {$\b$};
                    \xdef\c{\the\numexpr\c+1}
                }
                \draw [->]
                    (middle) +(0, \Y/5-29/10) coordinate (temp)
                    (left) +(0, \y/5-11/10) to [out=0, in=180] (temp)
                    +(0.2, 0) -- (-1.1, \Y/5-29/10)
                ;
            }
            \draw [rounded corners=0.1cm]
                (-1.1, -3.1) rectangle (1.1, 3.1)
                (0, 3.1) node [above] {$\sqrt\nu$ tests per batch}
                (1.1, 0)
                node [right, align=left] {$k \sqrt\nu$ \\ batches}
                (0, -3.1) node [below] {($k \nu$ tests in total)}
            ;
        \end{scope}
    \end{tikzpicture}
    \caption{
        Gacha GT (this work): reshape the $\nu$-bit phone number into a
        $\sqrt\nu \times \sqrt\nu$ square array; copy-and-past each row
        to a random batch of tests.  This combines ideas from Figures
        \ref{fig:horizontal} and \ref{fig:vertical}.
    }                                                     \label{fig:QR}
\end{figure*}

\subsection{An unconventional list-decodability challenge}
                                                      \label{sec:circle}

    In this subsection, we want to invent a code where a phone number
    corresponds to a message and a code symbol corresponds to a batch of
    tests.  This code should meet the following challenge.  All
    constants are retrospectively fitted so readers can easily grasp
    their magnitudes.
    
    \begin{challenge}                                 \label{cha:circle}
        We want a code with alphabet size
        $2^{6\sqrt\nu} = 2^{\Theta(\text{\#tests per batch})}$,
        block length $8 k \sqrt\nu =\Theta(\text{\#batches})$,
        and code dimension $\sqrt\nu/3 =
        \Theta(\log_{|\text{alphabet}|}(\text{\#people}))$.
        A corrupted word will be synthesized via the following random
        process (see also the depiction in Figure~\ref{fig:circle}).
        \begin{itemize}
            \itemsep=0ex
            \item Select $k$ random codewords to form an
                $8 k \sqrt\nu$-by-$k$ array; each column is a codeword.
            \item In each column, circle $4 \sqrt\nu =
                \Theta(\text{dimension})$ symbols randomly.
            \item If a row contains exactly one circle, erase all but
                the circled symbol.
            \item If a row contains zero or more than two circles, erase
                the entire row.
            \item Readout the symbols (including erasures) from top to
                bottom.
        \end{itemize}
        The goal is to quickly recover the $k$ random codewords
        (with very few FPs and FNs).
    \end{challenge}

    Readers will see the solution to Challenge~\ref{cha:circle} in the
    next subsection and how Gacha GT uses Challenge~\ref{cha:circle}
    as a subroutine in Sections \ref{sec:weight} and \ref{sec:count}.
    The remainder of this subsection is remarks.

    One might wonder if Challenge~\ref{cha:circle} can be solved by
    using known results on list-decodable codes in a black-box manner,
    But our setting is niche and restricted: Our noise rate
    is as high as $1 - \Theta(1/k)$ because each
    person participates in only $1/k$ of the batches.  This together
    with the desired code rate of $\Theta(1/k)$ forces the alphabet size
    in known constructions to be $2^{\Omega(k^2)}$, while we only allow
    $2^{\OO(\nu)}$.  We bypass this alphabet issue by not constructing
    list-decodable codes with algebraic, worst-case guarantees but
    focusing on the bare minimum of what the challenge needs: a
    probabilistically list-decodable code, leveraging the fact that our
    construction and perceived noise model both have randomness.
    
    One might also wonder if list-disjunct matrices can solve
    Challenge~\ref{cha:circle}.  To answer that, we remark that
    list-disjunct matrices are often paired with a decoder that
    maintains a description of a subset of people that \emph{must}
    contain all sick people; the decoder then updates the description to
    gradually reduce the size of the subset down to $k$.  In our noisy
    world, however, negative tests can no longer prove the participants
    healthy so such a decoder cannot narrow-down the set of possibly
    sick people, making it and disjunct matrices not applicable.

    Back to Challenge~\ref{cha:circle} itself.  Because each symbol is
    circled with probability $4/8k$, each circle will be the only circle
    in its row with probability $\approx e^{-4/8}$.  Thus, each column
    will contribute $\approx 4 \sqrt{\nu/e}$ symbols to the synthesized
    word.  Each symbol carries $6 \sqrt\nu$ bits of information so we
    are expecting $24 \nu / \sqrt e \approx 14.5 \nu$ bits per column.
    Since this is more than the $\nu$ bits we need, we can afford using
    some of the bits to ``give hints'' on how to cluster symbols coming
    from the same codeword together.
    
    The next subsection is dedicated to our solution via a
    non-traditional folding of Reed--Solomon codes where the $s$th
    symbol, $s \in [8 k \sqrt\nu]$, consists of the polynomial's
    evaluations at $b_0$ (a fixed point for all $s$) and at $p_s$
    (distinct points for various $s$).

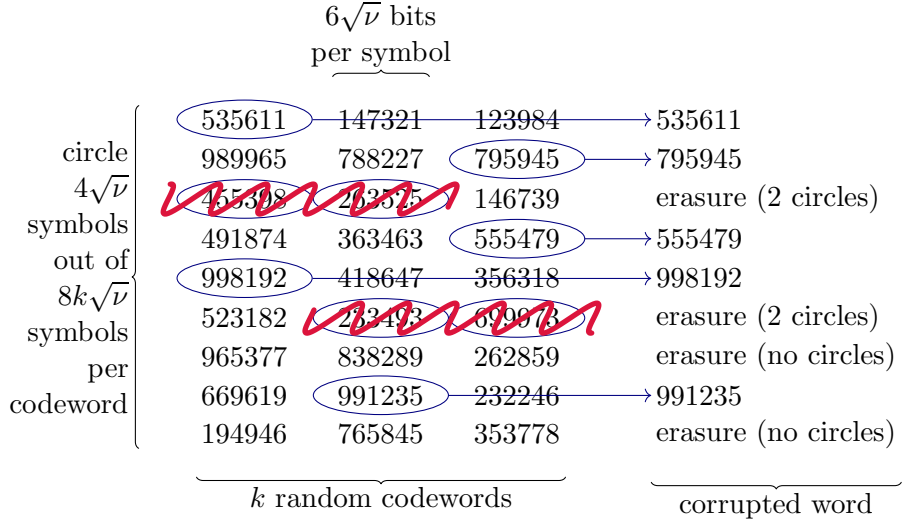
\begin{figure*}
    \centering
    \begin{tikzpicture}
        \def\2{erasure ($2$ circles)}
        \def\0{erasure (no circles)}
        \pgfmathsetseed{8881616}
        \def\r{%
            \pgfmathsetmacro\rr{random(9)}\rr
            \pgfmathsetmacro\rr{random(9)}\rr
            \pgfmathsetmacro\rr{random(9)}\rr
            \pgfmathsetmacro\rr{random(9)}\rr
            \pgfmathsetmacro\rr{random(9)}\rr
            \pgfmathsetmacro\rr{random(9)}\rr
        }
        \matrix [
            matrix of nodes,
            every cell/.style={inner sep=0.15em},
            column 5/.style=right,
            O/.style={ellipse, draw=Navy}
        ] (M) {
            |[O]| \r &       \r &       \r & $\kern2em$ & 535611 \\
                  \r &       \r & |[O]| \r & $\kern2em$ & 795945 \\
            |[O]| \r & |[O]| \r &       \r & $\kern2em$ & \2  \\
                  \r &       \r & |[O]| \r & $\kern2em$ & 555479 \\
            |[O]| \r &       \r &       \r & $\kern2em$ & 998192 \\
                  \r & |[O]| \r & |[O]| \r & $\kern2em$ & \2  \\
                  \r &       \r &       \r & $\kern2em$ & \0  \\
                  \r & |[O]| \r &       \r & $\kern2em$ & 991235 \\
                  \r &       \r &       \r & $\kern2em$ & \0  \\
        };
        \draw [Crimson, line width=2pt]
            [shift=(M-3-2)] plot [domain=-3:1, samples=200]
            ({\x + sin(\x*10 r)/5}, {sin(\x*10 r)/5})
            [shift=(M-6-2)] plot [domain=-1:3, samples=200]
            ({\x + sin(\x*10 r)/5}, {sin(\x*10 r)/5})
        ;
        \draw [Navy]
            (M-1-1) edge [->] (M-1-5)
            (M-2-3) edge [->] (M-2-5)
            (M-4-3) edge [->] (M-4-5)
            (M-5-1) edge [->] (M-5-5)
            (M-8-2) edge [->] (M-8-5)
        ;
        \draw [decorate, decoration=brace]
            (M-1-2.north west) +(0, 1em) coordinate (A)
            (M-1-2.north east) +(0, 1em) coordinate (B)
            (A) -- node [above, align=center]
            {$6\sqrt\nu$ bits \\ per symbol} (B)
        ;
        \draw [decorate, decoration=brace]
            (M-1-1.north west) +(-2em, 0) coordinate (A)
            (M-9-1.south west) +(-2em, 0) coordinate (B)
            (B) -- node [left, align=right] {
                circle \\ $4\sqrt\nu$ \\ symbols \\ out of \\
                $8 k \sqrt\nu$ \\ symbols \\ per \\ codeword
            }
            (A)
        ;
        \draw [decorate, decoration=brace]
            (M-9-1.south west) +(0, -1em) coordinate (A)
            (M-9-3.south east) +(0, -1em) coordinate (B)
            (B) -- node [below] {$k$ random codewords} (A)
        ;
        \draw [decorate, decoration=brace]
            (M-9-5.south west) +(0, -1em) coordinate (A)
            (M-9-5.south east) +(0, -1em) coordinate (B)
            (B) -- node [below] {corrupted word} (A)
        ;
    \end{tikzpicture}
    \caption{
        How Section~\ref{sec:circle} synthesizes corrupted words.
        Each $6$-digit number represents a symbol of $6\sqrt\nu$ bits.
        Each column is a codeword.  Each row in Figure~\ref{fig:QR}
        becomes a symbol here.
    }                                                 \label{fig:circle}
\end{figure*}

\subsection{A code design to meet Section~\ref{sec:circle}'s challenge}
                                                        \label{sec:fold}

    Let $\FF$ be a finite field of size $2^{3\sqrt\nu}$.  Let
    $\FF[t]^{<\sqrt\nu/3}$ be polynomials with degrees $< \sqrt\nu / 3$.
    There are $2^\nu = n$ such polynomials; each corresponds to a
    $\nu$-bit phone number and will evaluate into a codeword.  Pick $1 +
    8k \sqrt\nu$ evaluation points $b_0 \in \FF$ and $p_1, p_2, \dotsc,
    p_{8k\sqrt\nu} \in \FF$.  For $1 \leq j \leq 8 k \sqrt\nu$, let the
    $s$th symbol of the codeword be the evaluation pair $(g(b_0),
    g(p_s)) \in \FF^2$.  We call the first evaluation, $g(b_0)$, the
    \emph{birthday}.  The second evaluation, $g(p_s)$, is an coded
    fragment of the phone number.  See Figure~\ref{fig:fold} for some
    visualization.

    A comment on the phone--birthday analogy: We associate each person
    to a unique $\nu$-bit phone number because, in real world, most
    people have one phone and one number.  We call the evaluation
    $g(b_0)$ birthday because people like birthdays in their phone
    numbers; but it is mainly a reference to the birthday problem that
    addresses collision probabilities, which will play a crucial role in
    later proofs.

    A comparison to \cite{LNN19}'s approach: In their work, Larsen et
    al.\ synthesize the $s$th symbol as $(g(p_s); p_{s_1}, \dotsc,
    p_{s_d})$, where $s_1, \dotsc, s_d$ are the row indices of other $d$
    circles in the same column.  In other words, while our symbols obey
    a hash table structure $(\textsf{hash of data}, \textsf{data
    fragment})$, symbols of \cite{LNN19} obey a adjacency list structure
    $(\textsf{data fragment}; \textsf{pointers}, \textsf{to},
    \textsf{other}, \textsf{fragments})$.  Our data structure is
    superior in two ways.  (A) This section assumes $\log_2 k \leq
    \sqrt{\log_2 n}$, while \cite[Section~4]{LNN19} assumes $\epsilon >
    1 / \sqrt{\OO(\log n)}$, which is analogous to $k < \OO(\log n)$.
    (B) Larsen et al.\ need to perform a nontrivial \emph{clustering}
    algorithm, which is a nontrivially defined operation lying
    somewhere in between identifying connected components and
    identifying cliques.  On the other hand, Gacha simply uses
    the hash values to classify the data fragments.

\begin{figure*}
    \pgfmathdeclarefunction{g}{1}{\pgfmathparse{
        1 + sin(50*#1)/2 + sin(110*#1+20)/3 + sin(190*#1+70)/4
    }}
    \begin{tikzpicture} [baseline=0, x=0.8cm, y=0.8cm]
        \draw
            (4.5, 2.5) node {$\Bigl[
                \bigl(g(p_1), g(p_2)\bigr) \kern1em
                \bigl(g(p_3), g(p_4)\bigr) \kern1em
                \bigl(g(p_5), g(p_6)\bigr)
            \Bigr]$}
            plot [domain=0:9, samples=200] (\x, {g(\x)})
            node [left] {$g$}
            (0, 0) -- (9, 0)
            foreach [count=\i] \x in {1, 2, 4, 5, 7, 8}{
                (\x, 0) circle (1pt) node [below] {$p_\i$}
                (\x, {g(\x)}) circle(1pt)
                (\x, {g(\x)}) edge [->, dotted] (\x, 2.2)
            }
        ;
    \end{tikzpicture}
    \hfill
    \begin{tikzpicture} [baseline=0, x=0.8cm, y=0.8cm]
        \draw
            (4.5, 2.5) node {$\Bigl[
                \bigl(g(b_0), g(p_1)\bigr) \kern1em
                \bigl(g(b_0), g(p_2)\bigr) \kern1em
                \bigl(g(b_0), g(p_3)\bigr)
            \Bigr]$}
            plot [domain=0:9, samples=200] (\x, {g(\x)})
            node [left] {$g$}
            (0, 0) -- (9, 0)
            (1, 0) circle (1pt) node [below] {$b_0$}
            (1, {g(1)}) circle(1pt)
            foreach [count=\i] \x in {2, 5, 8}{
                (\x, 0) circle (1pt) node [below] {$p_\i$}
                (\x, {g(\x)}) circle(1pt)
                (1, {g(1)})
                edge [->, dotted, out=180-\i*60, in=-120] (\x-1.3, 2.2)
            }
        ;
    \end{tikzpicture}
    \caption{
        Left: folded Reed--solomon code in literature.
        Right: our customized Reed--Solomon code.
    }                                                   \label{fig:fold}
\end{figure*}
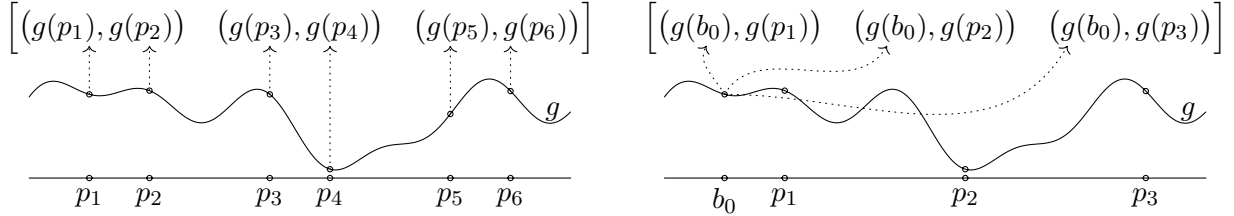

\subsection{How to list-decode Section~\ref{sec:fold}'s design}

    Given the design in Section~\ref{sec:fold}, synthesize a corrupted
    word as described in Section~\ref{sec:circle}.  A symbol of the
    synthesized word is either an erasure or an evaluation pair
    $(g(b_0), g(p_s)) \in  \FF^2$ for some unknown polynomial $g \in
    \FF[t]^{<\sqrt\nu/3}$.  We know that if the $s_1$th symbol and the
    $s_2$th symbol came from the same codeword, they will share the same
    birthday $g(b_0)$.  On the other hand, if two symbols came from two
    different polynomials $g$ and $h$, it is unlikely that $g(b_0)$
    equals $h(b_0)$.

    \begin{lemma} [Birthdays are collision-free]    \label{lem:birthday}
        Randomly select $k$ polynomials $K \subset
        \FF[t]^{<\sqrt\nu/3}$.  Then the expected number of polynomials
        having non-unique evaluations at $b_0$ (i.e., $g(b_0) = h(b_0)$
        for $g \neq h \in K$) is at most $2^{-\sqrt\nu-1}$.
    \end{lemma}

    \begin{proof}
        This is the birthday problem with $k \leq 2^{\sqrt\nu}$
        individuals and $|\FF| = 2^{3\sqrt\nu}$ days in a year.
    \end{proof}

    On the basis of Lemma~\ref{lem:birthday}, our decoding strategy is
    to sort the symbols of the synthesized word by birthday and then
    interpolate.  That is, suppose that the $s_1$th, the $s_2$th,
    \ldots, and the $s_{\sqrt\nu/3}$th symbols all share the same
    birthday.  We can recover $g$ by interpolating $(p_{s_1},
    g(p_{s_1}))$, $(p_{s_2}, g(p_{s_2}))$, etc.  To ensure that there
    are sufficiently many points to interpolate, we count the number of
    circles.

    \begin{lemma} [Evaluation pairs are ample]   \label{lem:concentrate}
        After executing the itemized process described in
        Section~\ref{sec:circle},  The probability that
        the first column has fewer than $\sqrt\nu / 3$
        circles left is $< 2^{-2\sqrt\nu-1}$.
    \end{lemma}

    \begin{proof}
        By symmetry, what applies to the first column applies to every.
        The first column starts with $4 \sqrt\nu$ circles.  A circle
        will be erased if another circle appears in the same row.  Since
        there are $8 k \sqrt\nu$ rows and $4 k \sqrt\nu$ circles, the
        probability of erasure is $< 4 k \sqrt\nu \div 8 k \sqrt\nu =
        1/2$.  In other words, whether or not a circle is kept is a
        Bernoulli trial with mean $> 1/2$, while we only need mean
        $\sqrt\nu / 3 \div 4 \sqrt\nu = 1/12$.  By Hoeffding's
        inequality, the probability of too many erasures is $\leq
        \exp(-2(1/2 - 1/12)^2 \cdot 4 \sqrt\nu) < 2^{-2.003\sqrt\nu}$.
        If $n$ is large enough this becomes $< 2^{-2\sqrt\nu-1}$.
    \end{proof}

    We can now analyze the success rate of our probabilistic
    list decoder.

    \begin{proposition} [Decoder is reliable]           \label{pro:fold}
        With $< 2^{-\sqrt\nu}$ FNs on average (over the random
        choice of polynomials and circles), the code defined in
        Section~\ref{sec:fold} can be list-decoded from the corrupted
        word synthesized by Section~\ref{sec:circle}.
    \end{proposition}

    \begin{proof}
        The proof amounts to analyzing what could go wrong.  One thing
        that could go wrong is that two polynomials have the same
        evaluation at $b_0$.  As addressed in Lemma~\ref{lem:birthday},
        there are $< 2^{-\sqrt\nu-1}$ bad polynomials on average.  The
        other thing that could go wrong is that some column does not
        contribute sufficiently many symbols to the synthesized word.
        As addressed in Lemma~\ref{lem:concentrate}, each column has
        probability $< 2^{-2\sqrt\nu-1}$ to do so.  So, the expected
        number of bad columns is $< k 2^{-2\sqrt\nu-1} \leq
        2^{-\sqrt\nu-1}$.  Adding the two $2^{-\sqrt\nu-1}$ concludes
        the proof of Proposition~\ref{pro:fold}.
    \end{proof}

\subsection{Apply the list-decoding idea to GT}       \label{sec:weight}

    The code defined in Section~\ref{sec:fold} has $\bigl|
    \FF[t]^{<\sqrt\nu/3} \bigr| = n$ codewords.  For each $j \in [n]$,
    we use the $j$th codeword to construct the $j$th column of the
    configuration matrix $A$ by the following procedure.
    \begin{itemize}
        \itemsep=0ex
        \item Out of the $8 k \sqrt\nu$ symbols of the $j$th codeword,
            we randomly select $3 \sqrt\nu$ symbols and replace the
            other symbols by a special symbol that represents erasure.
        \item Apply an injective map $\phi\colon \FF^2 \cup
            \{\text{erasure}\} \to \{0, 1\}^{7\sqrt\nu}$ to each symbol
            of the $j$th codeword.  This $\phi$ must\footnote{ Note that
            $|\FF^2| = 2^{6\sqrt\nu} < \binom{7\sqrt\nu}{3.5\sqrt\nu}$
            so $\phi$ can be chosen injective.} map the erasure symbol
            to the all zero string and map any evaluation pair $(g(b_0),
            g(p_s))$ to a binary string of Hamming weight $3.5
            \sqrt\nu$.
        \item Concatenate the $8 k \sqrt\nu$ images of $\phi$ to form a
            long string of length $56 k \nu$; this is the $j$th column
            of $A$.
    \end{itemize}
    
    In summary, we have described how to generate a random configuration
    matrix $A \in \{0, 1\}^{56 k \nu \times n}$.  We then sample a
    random $A$ and perform GT: $y \coloneqq \min(Ax, 1)$.

\subsection{How to decode Section~\ref{sec:weight}'s GT design}
                                                       \label{sec:count}
    
    Given $y = \min(Ax, 1)$, we want to find the set $J \subset [n]$
    such that $j \in J$ are those indices such that $x_j = 1$.  To do
    so, recall that $y_1$ to $y_{8\sqrt\nu}$ correspond to
    $\phi(\text{the 1st symbol})$, $y_{1+8\sqrt\nu}$ to $y_{16\sqrt\nu}$
    correspond to $\phi(\text{the 2nd symbol})$, and so on.  We denote
    the sub-vector corresponding to $\phi(\text{the $s$th symbol})$
    by $y'_s$.  Each $y'_s$ falls into one of the following three cases.
    \begin{itemize}
        \itemsep=0ex
        \item $y'_s$ is all zero: This means that for every selected
            codeword $j \in J$, the $s$th symbol is erased.
        \item $y'_s$ has more than $3.5 \log_2 k$ ones: Because $\phi$
            maps non-erasure symbols to constant-weight binary strings,
            there are at least two $j$'s that are not erased at this
            position.        
        \item $y'_s$ has exactly $3.5 \log_2 k$ ones: We apply the
            inverse of $\phi$ to recover the pair $(g(b_0),
            g(p_s))$.
            (We assume that the complexity of $\phi^{-1}$
            is negligible.)
    \end{itemize}
    Repeating this itemized procedure for all symbols, we will obtain a
    corrupted word that is the same as what Section~\ref{sec:circle}
    would have synthesized.  Using Proposition~\ref{pro:fold}, we
    conclude that $A$ forms a good GT scheme and proves a toy case of
    the main theorem.

    \begin{theorem} [Toy Gacha]                          \label{thm:toy}
        The configuration matrix $A$ described in
        Section~\ref{sec:weight} together with the decoding algorithm
        described in this subsection form a GT scheme.  This GT scheme
        uses $56 k \nu$ tests to find $k$ sick people in a population of
        $n$.  On average (over the randomness from $x$, $A$, and $\DD$),
        this scheme produces $2^{-\sqrt\nu}$ or fewer FNs (that is,
        $1/k$ or fewer).  Its decoding complexity is dominated by the
        complexity of interpolating $k$ polynomials in
        $\FF[t]^{<\sqrt\nu/3}$.
    \end{theorem}

    Theorem~\ref{thm:toy} captures the essence of the main theorem.  We
    name this scheme \emph{Gacha} after a genre of video games.  In a
    \emph{kompu gacha} game, a player spends money to obtain random
    coupons; these coupons can be exchanged for a treasure if the player
    has collected a full set \cite{STS19}.  In our case, tests are
    money, evaluations are coupons, and polynomials are treasures.

\section{Gadgets that Improve Gacha}                  \label{sec:gadget}

    In this section, we introduce three stackable gadgets that can
    improve any existing GT scheme.

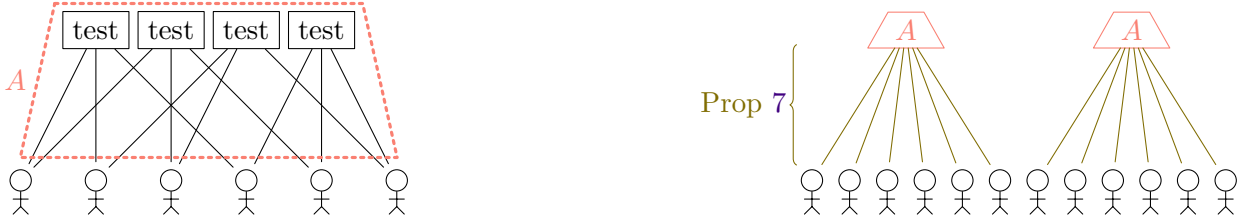
\begin{figure*}
    \begin{tikzpicture}
        \foreach \j in {1, ..., 6}{
            \draw
                (\j - 3.5, 0) node (p\j)
                [circle, draw, inner sep=0.1cm] {}
                (p\j.south) -- +(0, -0.2)
                +(-0.1, -0.1) -- +(0.1, -0.1)
                +(-0.1, -0.3) -- +(0, -0.2) -- +(0.1, -0.3)
            ;
        }
        \foreach \i in {1, ..., 4}{
            \draw (\i - 2.5, 2) node (t\i) [draw] {test} ;
        }
        \draw [shorten <=0.1cm]
            (p1) edge (t1) (p1) edge (t2)
            (p2) edge (t1) (p2) edge (t3)
            (p3) edge (t2) (p3) edge (t3)
            (p4) edge (t1) (p4) edge (t4)
            (p5) edge (t2) (p5) edge (t4)
            (p6) edge (t3) (p6) edge (t4)
        ;
        \draw [Salmon, dotted, line width=1.2pt]
            (t1.north west) +(-0.1, 0.1) coordinate (A)
            (t4.north east) +(0.1, 0.1) coordinate (B)
            (p6.north east) +(-0.1, 0.2) coordinate (C)
            (p1.north west) +(0.1, 0.2) coordinate (D)
            (A) -- (B) -- (C) -- (D) -- node [left] {$A$} cycle
        ;
    \end{tikzpicture}
    \hfill
    \begin{tikzpicture}
        \foreach \j in {1, ..., 12}{
            \draw ({(\j - 6.5) * 0.5}, 0)
                node (p\j) [circle, draw, inner sep=0.1cm] {}
                (p\j.south) -- +(0, -0.2)
                +(-0.1, -0.1) -- +(0.1, -0.1)
                +(-0.1, -0.3) -- +(0, -0.2) -- +(0.1, -0.3)
            ;
        }
        \foreach \i in {1, ..., 2}{
            \draw [Salmon] ({(\i - 1.5) * 3}, 2)
                node (t\i) [inner sep=0.1cm] {~~$A$~~~}
                (t\i.south west) -- (t\i.south east)
                -- (t\i.45) -- (t\i.135) -- cycle
            ;
            \foreach \j in {1, ..., 6}{
                \pgfmathtruncatemacro\J{\j + (\i - 1) * 6}
                \draw [Gold!50!black, shorten <=0.1cm]
                    (p\J) edge (t\i)
                ;
            }
        }
        \draw [Gold!50!black, decorate, decoration=brace]
            (p1.north west) ++(-0.1, 0.1) --
            node [left] {Prop~\ref{pro:parallel}} ++(0, 1.6)
        ;
    \end{tikzpicture}
    \caption{
        Left: a small GT scheme as a building block.  Right: a large GT
        scheme built by disconnected, independent copies of $A$.
    }                                               \label{fig:parallel}
\end{figure*}

\subsection{Repeat to increase throughput}

    In this subsection, we take an existing GT scheme (such as
    Theorem~\ref{thm:toy}) and repeat it multiple times for disjoint
    subsets of people.  Doing so yields free GT schemes in different
    parameter regions.

    \begin{proposition} [Repeat to increase throughput]
                                                    \label{pro:parallel}
        Assume that there is a GT scheme that uses $m$ tests and
        decoding complexity $d$ to find $k$ sick people in a population
        of $n$, and produces $f$ fps and fns on average.  Let $\varpi$
        be any positive integer.  There is a GT scheme using $\varpi m$
        tests and decoding complexity $\varpi d$ to find $\varpi k / 2$
        sick people in a population of $\varpi n$, and produces $\varpi
        f + 2^{-\Omega(k)} \varpi$ FPs and FNs on average.
    \end{proposition}

    \begin{proof}
        The idea is drawn in Figure~\ref{fig:parallel}.  We divide
        $\varpi n$ people into $\varpi$ groups, each containing $n$.
        We then apply $A$ to each group, where $A$ is the $m \times n$
        configuration matrix of the GT scheme that is promised by the
        statement of the proposition.  More formally, we construct a
        $\varpi m \times \varpi n$ block-diagonal matrix as
        \[
            \AA \coloneqq
            \bma{
                A
                \\ & A
                \\ & & \ddots
                \\ & & & A
            }
            \in \{0, 1\}^{\varpi m\times \varpi n}   
        \]

        To decode $\AA$, decode each $A$ separately.  By the promise,
        $A$'s decoding algorithm will produce $f$ fps and fns per group
        if this group contains fewer than $k$ sick people.  (Lowercase
        fp and fn are for false positive and false negative on the
        $A$-level.) If the group contains more than $k$ sick people, the
        decoder would not notice and would still attempt to generate a
        list of phone numbers.  It is safe to assume that this list has
        length $\leq 2k$ because any list longer than that contains
        $\geq k$ fps and can be replaced by an empty list that produces
        $k$ fns.  Therefore, a group oveflown by $\geq k$ people
        produces at most $3k$ fps and fns.

        The probability that a group contains more than $k$ sick people
        is $2^{-\Omega(k)}$ by applying Hoeffding's inequality to the
        fact that each group expects $k/2$ sick people.  Therefore,
        there will be $\varpi f + 2^{-\Omega(k)} \varpi \cdot 3k$ or
        fewer FPs and FNs on average.  But $2^{-\Omega(k)} \cdot 3k$ is
        just $2^{-\Omega(k)}$, finishing the proof.
    \end{proof}

    Proposition~\ref{pro:parallel} is a powerful tool in that if $\log_2
    n < \log_2(k)^2$, we can resize $k$ and $n$ until $\log_2 n =
    \log_2(k)^2$ holds.  Once the equality holds, we can apply
    Theorem~\ref{thm:toy}.  This makes Theorem~\ref{thm:toy} independent
    of the parameter regime and pushes the current bounds in
    partial-recovery GT and exact-recover GT in a unified manner.

    \begin{remark} [\textbf{Comparison with heavy hitters}]
        \def\abs#1{\mathopen|#1\mathclose|}
        There is a parallel between Proposition~\ref{pro:parallel} and
        \cite[Theorem~5]{LNN16a}: In the Appendix~A of the preprint
        \cite{LNN16a} of \cite{LNN19}, the authors proposed a reduction
        theorem from a small-tail heavy hitter problem to multiple
        large-tail heavy hitter problems.  For any $\epsilon, \delta \in
        (0, 1/2)$, their construction of an $\epsilon$-summarizer with
        failure probability $\delta$ is to prepare $\epsilon^{-2}
        \abs{\log_2 \delta}$ instances of $\OO(\abs{\log_2
        \delta}^{-1/2})$-summarizer, and hash every incoming item into
        one of the $\OO(\abs{\log_2 \delta}^{-1/2})$-summarizers.
        Modulo the choices of parameters, our
        Proposition~\ref{pro:parallel} shares the same idea as theirs.

        That said, \cite[Theorem~5]{LNN16a} allows $\delta$ to stay
        constant while $\epsilon \to 0$, meaning that there will be
        $\Theta(\epsilon^{-2}) \to \infty$ heavy hitters and
        $\epsilon^{-2} \abs{\log_2 \delta} = \Theta(\epsilon^{-2}) \to
        \infty$ instances of constant-summarizer.  Assuming that the
        incoming items are hashed randomly, the number of heavy hitters
        each instance receives is a random variable obeying a Poisson
        distribution with a constant intensity.  With a overwhelmingly
        high probability, one of the instance will receive $\omega(1)$
        heavy hitters and overflow.  This argument contradicts
        \cite[inequality~(2)]{LNN16a}, and hence their results are
        limited to the parameter region $\epsilon > 1 / \sqrt{\OO(\log
        n)}$.

        Here we are also facing a similar issue: We strongly advice
        against choosing $\varpi > 2^{\Omega(k)}$ because there will be
        $2^{-\Omega(k)} \varpi > 1$ overflowing groups that contain more
        than $k$ sick people.  These overflowing groups, with a
        overwhelmingly high probability, will produce a lot of fns.
    \end{remark}

\subsection{Repeat to double-check}

    In this subsection, we take an existing GT scheme, repeat it
    multiple times, and let them vote.

    \begin{proposition} [Repeat to double-check]      \label{pro:serial}
        Suppose that there is a GT scheme that uses $m$ tests and
        decoding complexity $d$ to find $k$ sick people in a population
        of $n$, and produces $f$ fps and fns on average.  Let $\sigma$
        be any positive integer.  There is a GT scheme that uses $\sigma
        m$ tests and decoding complexity $\sigma d + \sigma k
        \poly(\log(\sigma n))$ to find $k$ sick people in a population
        of $n$, and produces $k \sqrt{4f/k}^\sigma$ FPs and FNs on
        average.
    \end{proposition}

    \begin{proof}
        Let $A$ be the $m \times n$ configuration matrix of the GT
        scheme that supposedly exists.  We duplicate $A$ $\sigma$ times
        and shuffle the columns of each copy of $A$, i.e., shuffle the
        phone numbers.  (We assume that the complexity of looking-up
        phone numbers is negligible.) This way, we can say that every
        copy of $A$ functions independently.  We stack the copies of $A$
        on top of each other to yield a $\sigma m \times n$ matrix,
        $\AA$.  For the formality,
        \[
            \AA \coloneqq
            \bma{
                A P_1
                \\ A P_2
                \\[-3pt] \vdots
                \\ A P_\sigma
            }
            \in \{0, 1\}^{\sigma m\times n},
        \]
        where $P_1, P_2, \dotsc, P_\sigma$ are $n \times n$ permutation
        matrices sampled uniformly at random.  We use $\AA$ to perform
        GT: $y \coloneqq \min(\AA x, 1)$.

        To decode $\AA$, decode each copy of $A$ separately.  Each $A$
        generates a list of $k \pm f$ or so phone numbers who are
        possibly sick.  We can guarantee that the length of this list
        never go beyond $2k$ because an empty list produces fewer fns
        than a $2k$-long list produces fps.  In total, there will be
        $\sigma k \pm \sigma f$ or so phone numbers, and is guarantee to
        be no more than $2 \sigma k$.  We generate the final list of
        sick people by collecting phone numbers that appear at least
        $\sigma/2$ times.  From what is above, the decoding complexity
        is $\sigma d$ for separately decoding all copies of $A$ plus
        $\OO(k' \log k')$, where $k' \coloneqq 2 \sigma k$, for sorting
        the phone numbers.

        Now we count the number of FNs and FPs.  A sick person becomes
        an FN if her phone number appears $< \sigma/2$ times.  Each copy
        of $A$ fails to find this sick person with probability
        $\varepsilon \coloneqq f/k$.  The probability that $A$ fails
        $\sigma/2$ times is, by the union bound, $<
        \binom{\sigma}{\sigma/2} \varepsilon^{\sigma/2} < 0.5
        \sqrt{4\varepsilon}^\sigma$.  We next bound the number of FPs.
        A healthy person becomes an FP if her phone number appears $>
        \sigma/2$ times.  But each copy of $A$ incorrectly generates her
        phone number with probability $f/n < f/k = \varepsilon$ and
        hence the same upper bound, $0.5 \sqrt{4\varepsilon}^\sigma$,
        applies.  In sum, the total number of misclassified people will
        be $k \sqrt{4\varepsilon}^\sigma$ on average.
    \end{proof}

    The utility of Proposition~\ref{pro:serial} is, when paying a scalar
    of $\sigma$, to reduce the failure probability not by a scalar, but
    by raising it to the power of $\sigma$.  On a higher level,
    Proposition~\ref{pro:serial} facilitates a trade-off between $m$ and
    $f$ in a black-box way.  For example, if we apply the proposition to
    bit-mixing coding \cite{BCS21}, we obtain the same trade-off
    provided by the noisy splitting work \cite{PST23}.

    One of the core lessons of coding theory is to introduce repetition
    to improve the reliability of some unreliable build and then
    introduce codes to reduce the cost of repetition.  In the following,
    we will see how Proposition~\ref{pro:serial} is codified.

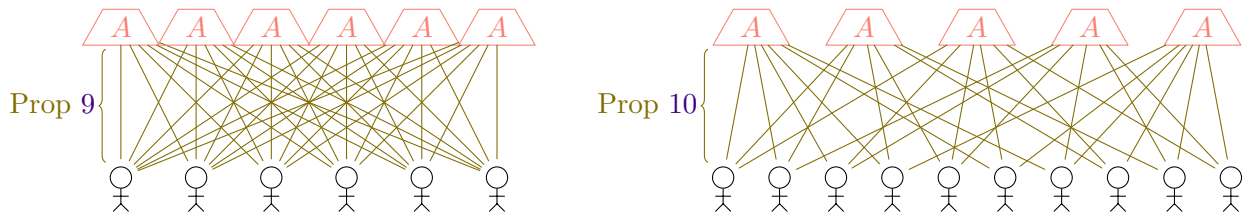
\begin{figure*}
    \begin{tikzpicture}
        \foreach \j in {1, ..., 6}{
            \draw (\j - 3.5, 0)
                node (p\j) [circle, draw, inner sep=0.1cm] {}
                (p\j.south) -- +(0, -0.2)
                +(-0.1, -0.1) -- +(0.1, -0.1)
                +(-0.1, -0.3) -- +(0, -0.2) -- +(0.1, -0.3)
            ;
        }
        \foreach \i in {1, ..., 6}{
            \draw [Salmon] (\i - 3.5, 2)
                node (t\i) [inner sep=0.1cm] {~~$A$~~~}
                (t\i.south west) -- (t\i.south east)
                -- (t\i.45) -- (t\i.135) -- cycle
            ;
            \foreach \j in {1, ..., 6}{
                \draw [Gold!50!black, shorten <=0.1cm]
                    (p\j) edge (t\i)
                ;
            }
        }
        \draw [Gold!50!black, decorate, decoration=brace]
            (p1.north west) ++(-0.1, 0.1) --
            node [left] {Prop~\ref{pro:serial}} ++(0, 1.5)
        ;
    \end{tikzpicture}
    \hskip-1cm plus1fill
    \begin{tikzpicture}
        \foreach \j in {1, ..., 10}{
            \draw ({(\j - 5.5) * 0.75}, 0)
                node (p\j) [circle, draw, inner sep=0.1cm] {}
                (p\j.south) -- +(0, -0.2)
                +(-0.1, -0.1) -- +(0.1, -0.1)
                +(-0.1, -0.3) -- +(0, -0.2) -- +(0.1, -0.3)
            ;
        }
        \foreach \i in {1, ..., 5}{
            \draw [Salmon] ({(\i - 3) * 1.5}, 2)
                node (t\i) [inner sep=0.1cm] {~~$A$~~~}
                (t\i.south west) -- (t\i.south east)
                -- (t\i.45) -- (t\i.135) -- cycle
            ;
        }
        \draw [Gold!50!black, shorten <=0.1cm]
            (p1) edge (t1) (p1) edge (t2) (p1) edge (t3)
            (p2) edge (t1) (p2) edge (t2) (p2) edge (t4)
            (p3) edge (t1) (p3) edge (t3) (p3) edge (t4)
            (p4) edge (t1) (p4) edge (t2) (p4) edge (t5)
            (p5) edge (t2) (p5) edge (t3) (p5) edge (t4)
            (p6) edge (t1) (p6) edge (t3) (p6) edge (t5)
            (p7) edge (t1) (p7) edge (t4) (p7) edge (t5)
            (p8) edge (t2) (p8) edge (t3) (p8) edge (t5)
            (p9) edge (t2) (p9) edge (t4) (p9) edge (t5)
            (p10) edge (t3) (p10) edge (t4) (p10) edge (t5)
        ;
        \draw [Gold!50!black, decorate, decoration=brace]
            (p1.north west) ++(-0.1, 0.1) --
            node [left] {Prop~\ref{pro:expander}} ++(0, 1.5)
        ;
    \end{tikzpicture}
    \caption{
        Left: Proposition~\ref{pro:serial} can be seen as a construction
        based on a complete bipartite graph.  Right:
        Proposition~\ref{pro:expander} can be seen as a construction
        based on a ``good'' graph.
    }                                                 \label{fig:serial}
\end{figure*}

\subsection{Repeat to double-check increased throughput}

    We recognize that Proposition~\ref{pro:parallel} can be visualized
    as a collection of claws and Proposition~\ref{pro:serial} can be
    visualized as a complete bipartite graph.  In this subsection, we
    generalize the propositions to generic graphs and prove the
    following.

    \begin{proposition} [Repeat to double-check increased throughput]
                                                    \label{pro:expander}
        Suppose that there is a GT scheme that uses $m$ tests and
        decoding complexity $d$ to find $k$ sick people in a population
        of $n$, and produces $f$ fns and fps on average.  Let $\rho$ and
        $\varrho$ be free parameters, $1 < \rho < \varrho < \sqrt n$.
        There is a GT scheme using $M \coloneqq \varrho m$ tests and
        decoding complexity $D \coloneqq \varrho d + \varrho k
        \poly(\rho\log(\varrho n))$ to find $K \coloneqq \varrho k / 2
        \rho$ sick people in a population of $N \coloneqq n^{\rho/6}$,
        and produces $F \coloneqq 2^\rho (f/k + 2^{-\Omega(k)})^{\rho/3}
        K + K^2 / 2 \sqrt n + \varrho^{\rho/3} (f +
        2^{-\Omega(k)})^{\rho/3} \sqrt n^{1-\rho/3}$ FNs and FPs.
    \end{proposition}
    
    Although the statement is cryptic, Proposition~\ref{pro:expander} is
    a simple, natural generalization of Proposition~\ref{pro:serial}
    in the following manner: We repeat the configuration matrix
    $A$ $\varrho$ times, and ask every person to participate in $\rho$
    random copies of $A$, as shown in Figure~\ref{fig:serial}.  In
    doing so, each person encodes her phone number $j \in [N]$ into a
    tuple $(j_1, \dotsc, j_\varrho) \in [n]^\varrho$, and pretends that
    her phone number is $j_r \in [n]$ when participating in the $r$th
    copy of $A$.

    Given the design philosophy in the last paragraph, let us explain
    the numerals $\varrho k / 2 \rho$ and $n^{\rho/6}$.  For why the
    number of sick people should be $K \coloneqq \varrho k / 2 \rho$,
    consider a random bipartite graph that consists of the people part
    and the copies-of-$A$ part (cf.\  Figure~\ref{fig:serial}).  In this
    graph, $\varrho$ is the number of $A$-vertices; $\rho$ is the degree
    of people-vertices; and $k$ is the maximal number of sick people
    that each copy of $A$ can handle.  For ``safety'', we choose the
    parameters such that the expected number of sick people connected to
    each copy of $A$ is $k/2$.  By edge-counting, the number of sick
    people we can safely find is about $\varrho k / 2 \rho \eqqcolon K$.

    For why the population should be $N \coloneqq n^{\rho/6}$, consider
    the following back-of-the-envelop computation.  Out of the $\rho$
    copies of $A$ a sick person participates in, we hope that $\rho/3$
    or more copies can decode successfully.\footnote{ Any fraction can
    replace $1/6$; for instance, Lemma~\ref{lem:concentrate} has
    $\sqrt\nu / 3 \div 4\sqrt\nu = 1/12$.} Because each copy of $A$
    provides $\nu = \log_2 n$ bits of information, we expect receiving
    $\rho \nu / 3$ bits of information regarding the identity of each
    sick person.  Reusing the idea of Section~\ref{sec:fold}, we use
    half of the bits ($\rho \nu / 6$) to encode birthday and the other
    half to encode phone number.  Hence, we can handle a population of
    $n^{\rho/6} \eqqcolon N$.

    We can now present the formal proof.

    \begin{proof} [Proof of Proposition~\ref{pro:expander}]
        Let there be a finite field $\FF$ of size $\sqrt n$.  Let
        $\FF[t]^{<\rho/3}$ be the set of $N = n^{\rho/6}$ polynomials
        whose degrees are $< \rho/3$.  Pick $1 + \varrho$ evaluation
        points $b_0, p_1, \dotsc, p_\varrho \in \FF$.

        For each $j \in [N]$, assign a unique polynomial $g \in
        \FF[t]^{<\rho/3}$ to the $j$th person.  We now construct the
        $j$th column of the configuration matrix $\AA \in \{0,
        1\}^{Rm\times N}$ using $g$.  For any $r \in [\varrho]$, if the
        $j$th person chooses not to participate in the $r$th copy of
        $A$, let the matrix entries $\AA_{(rm-m+1,j)}, \dotsc,
        \AA_{(rm,j)}$ be zero.  Otherwise, if the $j$th person does
        choose to participate in the $r$th copy of $A$, she will compute
        $(g(b_0), g(p_r)) \in \FF^2$ and use some map $\varphi\colon
        \FF^2 \to [n]$ to turn the $r$th evaluation pair to an integer.
        She then copies the $\varphi(g(b_0), g(p_r))$th column of $A$
        and pastes them onto $\AA_{(rm-m+1,j)}, \dotsc, \AA_{(rm,j)}$.

        To decode $\AA$, decode each copy of $A$ separately.  Each $A$
        will return a subset of roughly $k \pm f$ (never more than $2k$)
        indices in $[n]$.  Apply $\varphi^{-1}$ to turn the indices to a
        subset of evaluation pairs of the form $(g(b_0), g(p_r))$.  (We
        assume that the complexity of $\varphi^{-1}$ is negligible.)
        Sort these $\varrho (k \pm f)$ pairs by birthday and ignore
        birthdays appearing fewer than $2\rho/3$ times.  Finally, for
        each birthday with $2\rho/3$ evaluations, interpolate the
        evaluations as if we are decoding a Reed--Solomon code.

        The decoding complexity is the sum of decoding copies of $A$
        (i.e., $\varrho d$), sorting birthdays ($\OO(k' \log k')$, where
        $k' \coloneqq 2 \varrho k$), and interpolating polynomials
        ($\OO(\varrho k \rho^3)$).
        
        We now count FNs.  There are two sources producing FNs.  The
        first source is sick people who share birthday with others.
        They contribute $K^2 / 2 \sqrt n$ FNs on average; cf.\
        Lemma~\ref{lem:birthday}.  The second source is sick people who
        cannot be decoded by more than $\rho/3$ copies of $A$; they
        contribute $2^\rho (f/k + 2^{-\Omega(k)})^{\rho/3} K$ FNs.
        Here, the term $2^{-\Omega(k)}$ corresponds to the event that a
        copy of $A$ is overflown by $> k$ sick people; cf.\ the proof of
        Proposition~\ref{pro:parallel}.  The term $4f/k$ corresponds to
        the fns $A$ makes when $A$ is not overflown; cf.\ the proof of
        Proposition~\ref{pro:serial}.

        We now count FPs.  The only source of FP is that some subset of
        $\rho/3$ copies of $A$ produce fps sharing the same birthday.
        For each of the $\sqrt n$ birthdays, there are $f / \sqrt n$ fps
        sharing that birthday per copy of $A$; there will be more fps
        if $A$ overflows but that happens with small probability
        $2^{-\Omega(k)}$.  There are $\binom\varrho{\rho/3} \leq
        \varrho^{\rho/3}$ subsets of copies of $A$.  Therefore,
        the number of FPs is fewer than $\varrho^{\rho/3}
        (f + 2^{-\Omega(k)})^{\rho/3} \sqrt n^{1-\rho/3}$.
        This finishes the proof.
    \end{proof}

    In the next subsection, we will see that
    Proposition~\ref{pro:expander} generalizes Theorem~\ref{thm:toy} in
    the way list-recoverable codes generalize list-decodable code.
    Although some earlier works had use list-recoverable codes to piece
    together phone numbers, Proposition~\ref{pro:expander} sits in an
    untraditional parameter corner so we end up constructing our own
    code.

\begin{figure*}
    \centering
    \begin{tikzpicture} [3d view={140}{15}, scale=0.6, overlay]
        \draw
            (0, 0) coordinate (A)
            (0, -1, -1) coordinate (B)
            (1, 0, -1) coordinate (C)
            (0, 1, -1) coordinate (D)
            (-1, 0, -1) coordinate (E)
            (-1.4, -0.3, -1) coordinate (E')
        ;
    \end{tikzpicture}
    \begin{tikzpicture} [baseline=($2.5/2*(B)+2.5/2*(D)$)]
        \begin{scope} [transform canvas={scale=2.5}]
            \fill [Gold!30!gray] (B) -- (D) -- (E');
            \fill [Gold!30!gray!50] (A) -- (D) -- (E);
            \fill [Gold!30] (A) -- (B) -- (C) -- (D);
        \end{scope}
        \draw
            (A) node [above, align=center] {
                Theorem~\ref{thm:toy}
                \\ $k \exp(-\Omega(\log n)^{1/2})$
                \\ FNs and FPs
            }
            ($1.5*(C)$) -- node [above] {$\phi$}
            ($1.5*(D)$) -- ($1.5*(E)$)
            ($2.5*(C)$) -- node [above] {Prop~\ref{pro:expander}}
            ($2.5*(D)$) -- ($2.5*(E)$)
            (A) -- ($2.5*(C)$) (A) -- ($2.5*(D)$) (A) -- ($2.5*(E)$)
        ;
    \end{tikzpicture}
    \hfil
    \begin{tikzpicture} [baseline=($3.5/2*(B)+3.5/2*(D)$)]
        \begin{scope} [transform canvas={scale=3.5}]
            \fill [Gold!30!gray] (B) -- (D) -- (E');
            \fill [Gold!30!gray!50] (A) -- (D) -- (E);
            \fill [Gold!30] (A) -- (B) -- (C) -- (D);
        \end{scope}
        \draw
            (A) node [above, align=center] {
                Theorem~\ref{thm:three}
                \\ $K \exp(-\Omega(\log N)^{2/3})$
                \\ FNs and FPs
            }
            ($1.5*(C)$) -- node [above] {$\phi$}
            ($1.5*(D)$) -- ($1.5*(E)$)
            ($2.5*(C)$) -- node [above] {Prop~\ref{pro:expander}}
            ($2.5*(D)$) -- ($2.5*(E)$)
            ($3.5*(C)$) -- node [above] {Proposit~\ref{pro:expander}}
            ($3.5*(D)$) -- ($3.5*(E)$)
            (A) -- ($3.5*(C)$) (A) -- ($3.5*(D)$) (A) -- ($3.5*(E)$)
        ;
    \end{tikzpicture}
    \hfil
    \begin{tikzpicture} [baseline=($4.5/2*(B)+4.5/2*(D)$)]
        \begin{scope} [transform canvas={scale=4.5}]
            \fill [Gold!30!gray] (B) -- (D) -- (E');
            \fill [Gold!30!gray!50] (A) -- (D) -- (E);
            \fill [Gold!30] (A) -- (B) -- (C) -- (D);
        \end{scope}
        \draw
            (A) node [above, align=center]
                {$K \exp(-\Omega(\log N)^{3/4})$ \\ FNs and FPs}
            ($1.5*(C)$) -- node [above] {$\phi$}
            ($1.5*(D)$) -- ($1.5*(E)$)
            ($2.5*(C)$) -- node [above] {Prop~\ref{pro:expander}}
            ($2.5*(D)$) -- ($2.5*(E)$)
            ($3.5*(C)$) -- node [above] {Proposit~\ref{pro:expander}}
            ($3.5*(D)$) -- ($3.5*(E)$)
            ($4.5*(C)$) -- node [above] {Proposition~\ref{pro:expander}}
            ($4.5*(D)$) -- ($4.5*(E)$)
            (A) -- ($4.5*(C)$) (A) -- ($4.5*(D)$) (A) -- ($4.5*(E)$)
        ;
    \end{tikzpicture}
    \caption{
        Proposition~\ref{pro:expander} is a stackable gadget that
        recursively evolves the parameters $(m, d, n, k, f)$.  By
        staking more and more layers of Proposition~\ref{pro:expander},
        we can improve the exponent of the exponent from $1/2$ to $2/3$
        to $3/4$ and all the way to $1 - 1/\tau$ (when there are $\tau -
        1$ layers).
    }                                                \label{fig:pyramid}
\end{figure*}
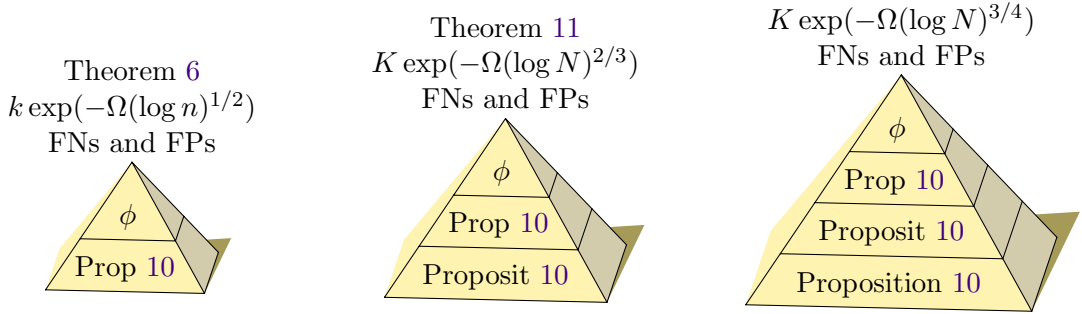

\subsection{Pyramids of stacked gadgets}

    In this subsection, we demonstrate that
    Proposition~\ref{pro:expander} is a stackable gadget that can evolve
    the parameters $(m, d, k, n, f)$ recursively.  That and Propositions
    \ref{pro:parallel} and \ref{pro:serial} stacked together can prove
    the noiseless version of the main theorem, as
    Figure~\ref{fig:pyramid} illustrates.  As a steeping stone toward
    the immense potential of Proposition~\ref{pro:expander}, we reprove
    Theorem~\ref{thm:toy} as the following.
    
    \bgroup
    \def\thetheorem{\addtocounter{theorem}{-1}}
    \begin{theorem} [Restating Theorem~\ref{thm:toy}
                     on page~\pageref{thm:toy}
                     with capital letters for readers' convenience]
        There is a GT scheme that uses $m = \OO(K \log N)$ tests and
        decoding complexity $O(K \poly(\log N))$ to find $K$ sick people
        in a population of $N$.  On average (over the randomness from
        $x$, $A$, $Z$, and $\DD$), it will produce $K
        \exp(-\Omega(\sqrt{\log N}))$ or fewer FPs and FNs.
    \end{theorem}
    \egroup

    \begin{proof}[Alternative proof of Theorem~\ref{thm:toy}]
        Let $\ell$ be a free, large parameter.  We focus on the case
        $(K, N) = (2^\ell, 2^{6\ell^2})$ to demonstrate the evolution of
        parameters better; other pairs of $(K, N)$ can be obtained by
        fine-tuning the selection of $(\rho, \varrho)$ and/or applying
        Proposition~\ref{pro:parallel}.

        To begin, find constant-weight code $\phi$ that has length
        $12\ell$ and $2^{6\ell}$ codewords.  Treat $\phi$ as a GT scheme
        that uses $12\ell$ tests and decoding complexity $(12\ell)^c$
        ($c$ is a very large constant) to find $1$ sick person in a
        population of $2^{6\ell}$.  Now apply
        Proposition~\ref{pro:expander} to $\phi$ with parameters
        \begin{center}
            \begin{tabular}{ccccccc}
                \toprule
                $m$ & $d$ & $k$ & $n$ & $f$ & $\rho$ & $\varrho$
                \\ \midrule
                $12\ell$ & $(12\ell)^c$ & $1$ & $2^{6\ell}$
                & $0$ & $6\ell$ & $2\rho 2^\ell / k$
                \\ \bottomrule
            \end{tabular}
        \end{center}

        We obtain capital-letter parameters
        \begin{center}
            \begin{tabular}{ccccc}
                \toprule
                $M$ & $D$ & $K$ & $N$ & $F$
                \\ \midrule
                $\varrho m$
                & $\varrho d$ + $\varrho k \poly(\rho \log(\varrho n))$
                & $\varrho k / 2\rho$ & $n^{\rho/6}$
                & $\substack{
                    2^\rho (f/k + 2^{-\Omega(k)})^{\rho/3} K
                    + K^2 / 2\sqrt n
                    \\ + \varrho^{\rho/3}
                    (f + 2^{-\Omega(k)})^{\rho/3} \sqrt n^{1-\rho/3}
                }$
                \\ $2^\ell (12\ell)^2$ & $2^\ell (12\ell)^{2c}$
                & $2^\ell$ & $2^{6\ell^2}$ & $2^{-\ell}$
                \\ \bottomrule
            \end{tabular}
        \end{center}
        Note that, when filling-in the table above, we treat
        $2^{-\Omega(k)}$ as zero because that term corresponds to the
        event that an $A$ overflows and produces a lot of fps and fns.
        With $\phi$, however, we can always detect overflow by checking
        if the Hamming weight matches.
        
        Now we observe that the number of tests is $M = O(K \log N)$,
        that the decoding complexity is $D = K \poly(\log N)$, and that
        the number of mistakes is $F = K \exp(-\Omega(\sqrt{\log N}))$,
        we conclude the proof of the case $\log_2 N = 6 \log_2(K)^2$.
        For the $>$ case, use a slightly different tuple of parameters;
        for the $<$ case, apply Proposition~\ref{pro:parallel}.
    \end{proof}

    The advantage of Proposition~\ref{pro:expander} is its stackability:
    Having seen that Theorem~\ref{thm:toy} is
    Proposition~\ref{pro:expander} applied to $\phi$, let us apply
    Proposition~\ref{pro:expander} to Theorem~\ref{thm:toy} with
    lowercase parameters
    \begin{center}
        \begin{tabular}{ccccccc}
            \toprule
            $m$ & $d$ & $k$ & $n$ & $f$ & $\rho$ & $\varrho$
            \\ \midrule
            $k (12\ell)^2$ & $k (12\ell)^{2c}$
            & $2^\ell$ & $2^{6\ell^2}$ & $2^{-\ell}$
            & $6\ell$ & $2\rho 2^{\ell^2} / k$
            \\ \bottomrule
        \end{tabular}
    \end{center}
    We will obtain capital parameters
    \begin{center}
        \begin{tabular}{ccccc}
            \toprule
            $M$ & $D$ & $K$ & $N$ & $F$
            \\ \midrule
            $2^{\ell^2} (12\ell)^3$ & $2^{\ell^2} (12\ell)^{3c}$
            & $2^{\ell^2}$ & $2^{6\ell^3}$ & $2^{-\ell^2}$
            \\ \bottomrule
        \end{tabular}
    \end{center}
    Using Proposition~\ref{pro:parallel} to handle the case $\log_2(K) >
    \log_2(N)^{2/3}$, we can prove this.

    \begin{theorem} [Stacked Gacha]                    \label{thm:three}
        There is a GT scheme that uses $\OO(K \log N)$ tests and
        decoding complexity $\OO(K \poly(\log N))$ to find $K$ sick
        people in a population of $N$.  On average (over the randomness
        from $x$, $A$, and $\DD$), it will produce $K \exp(-\Omega(\log
        N)^{2/3})$ or fewer FPs and FNs.
    \end{theorem}

    Applying Proposition~\ref{pro:expander} to Theorem~\ref{thm:three},
    we can evolve lowercase parameters
    \begin{center}
        \begin{tabular}{ccccccc}
            \toprule
            $m$ & $d$ & $k$ & $n$ & $f$ & $\rho$ & $\varrho$
            \\ \midrule
            $k (12\ell)^3$ & $k (12\ell)^{3c}$ & $2^{\ell^2}$
            & $2^{6\ell^3}$ & $2^{-\ell^2}$
            & $6\ell$ & $2\rho 2^{\ell^3} / k$
            \\ \bottomrule
        \end{tabular}
    \end{center}
    to capital parameters
    \begin{center}
        \begin{tabular}{ccccc}
            \toprule
            $M$ & $D$ & $K$ & $N$ & $F$
            \\ \midrule
            $2^{\ell^3} (12\ell)^4$ & $2^{\ell^3} (12\ell)^{4c}$
            & $2^{v^3}$ & $2^{6\ell^4}$ & $2^{-\ell^3}$
            \\ \bottomrule
        \end{tabular}
    \end{center}
    This further reduces the expected number of FNs and FPs to
    $K \exp(-\Omega(\log N)^{3/4})$.  In general, applying
    Proposition~\ref{pro:expander} $\tau - 1$ times increases the number
    of tests and complexity by $2^{\OO(\tau)}$-fold and reduces the
    expected number of FNs and FPs to $K \exp(-\log(N)^{1-1/\tau}
    2^{-\OO(\tau)})$.

    We now have all the tools to prove the main theorem for
    the case that tests are noiseless.  Consider the pyramid in
    Figure~\ref{fig:evolve}.  In this pyramid, people are at the base
    and the tests are at the tip.  The encoding process starts with
    Propositions \ref{pro:parallel} and \ref{pro:serial} that rescale
    $(m, d, k, n, f)$ by $\varpi$ and $\sigma$.  These propositions will
    produce ``virtual phone numbers'' in the intermediate layer.  We
    than proceed to stacking Proposition~\ref{pro:expander} to generate
    more virtual phone numbers in higher and higher layers (repeat this
    $\tau - 1$ times).  This encoding process will work all the way up
    until, finally, we apply $\phi$ to the virtual test phone numbers to
    generate the set of tests each person should participate in.

    To decode, apply $\phi^{-1}$ and work top-down: For every layer of
    the pyramid, we are given a collection of evaluation pairs $(g(b_0),
    g(p_s))$.  We sort by birthdays and interpolate the evaluations to
    recover the virtual phone numbers who are sick, and pass that
    information to the layer below.  Finally, at the bottom layer,
    we will obtain the phone numbers of the actual sick people.

    To prove the main theorem, it remains to address how to deal with
    noisy tests.

\begin{figure*}
    \centering
    \begin{tikzpicture}
        [baseline={([yshift=-1ex]current bounding box.center)}]
        \begin{scope} [transform canvas={scale=5.5}]
            \fill [Gold!30!gray] (B) -- (D) -- (E');
            \fill [Gold!30!gray!50] (A) -- (D) -- (E);
            \fill [Gold!30] (A) -- (B) -- (C) -- (D);
        \end{scope}
        \draw
            (A) node [above]{Main theorem but noiseless}
            ($1.5*(C)$) -- node [above] {$\phi$}
            ($1.5*(D)$) -- ($1.5*(E)$)
            ($2.5*(C)$) -- node [above] {Prop~\ref{pro:expander}}
            ($2.5*(D)$) -- ($2.5*(E)$)
            ($3.5*(C)$) -- node [above] {$\vdots$}
            ($3.5*(D)$) -- ($3.5*(E)$)
            ($4.5*(C)$) -- node [above] {Proposition~\ref{pro:expander}}
            ($4.5*(D)$) -- ($4.5*(E)$)
            ($5.5*(C)$) -- node [above]
                {Propositions \ref{pro:parallel} and \ref{pro:serial}}
            ($5.5*(D)$) -- ($5.5*(E)$)
            (A) -- ($5.5*(C)$) (A) -- ($5.5*(D)$) (A) -- ($5.5*(E)$)
        ;
    \end{tikzpicture}
    \hfill
    \def\tabcolsep{5pt}
    \begin{tabular}{ccccccc}
        \toprule

        $m$ & $d$ & $k$ & $n$ & $f$ & $\rho$ & $\varrho$

        \\ \midrule

        $12\ell$ & $(12\ell)^c$ & $1$ & $2^{6\ell}$
        & $0$ & $6\ell$ & $2\rho 2^\ell / k$

        \\ $k (12\ell)^2$ & $k (12\ell)^{2c}$
        & $2^\ell$ & $2^{6\ell^2}$ & $2^{-\ell}$
        & $6\ell$ & $2\rho 2^{\ell^2} / k$

        \\ $k (12\ell)^3$ & $k (12\ell)^{3c}$ & $2^{\ell^2}$
        & $2^{6\ell^3}$ & $2^{-\ell^2}$
        & $6\ell$ & $2\rho 2^{\ell^3} / k$

        \\[0pt] $\vdots$ & $\vdots$ & $\vdots$
        & $\vdots$ & $\vdots$ & $\vdots$ & $\vdots$
        \\[2pt]

        $k (12\ell)^{\tau-1}$ & $k (12\ell)^{(\tau-1)c}$
        & $2^{\ell^{\tau-2}}$ & $2^{6\ell^{\tau-1}}$
        & $2^{-\ell^{\tau-2}}$ & $6\ell$ & $2\rho 2^{\ell^{\tau-1}} / k$
        
        \\ $k (12\ell)^{\tau}$ & $k (12\ell)^{\tau c}$
        & $2^{\ell^{\tau-1}}$ & $2^{6\ell^\tau}$ & $2^{-\ell^{\tau-1}}$
        & --- & ---

        \\ \bottomrule
    \end{tabular}
    \caption{
        Left: a visualization of stacked Proposition~\ref{pro:expander}
        (along with other building blocks).  Right: how parameters
        evolve under Proposition~\ref{pro:expander} and prove the
        noiseless version of the main theorem.
    }                                                 \label{fig:evolve}
\end{figure*}
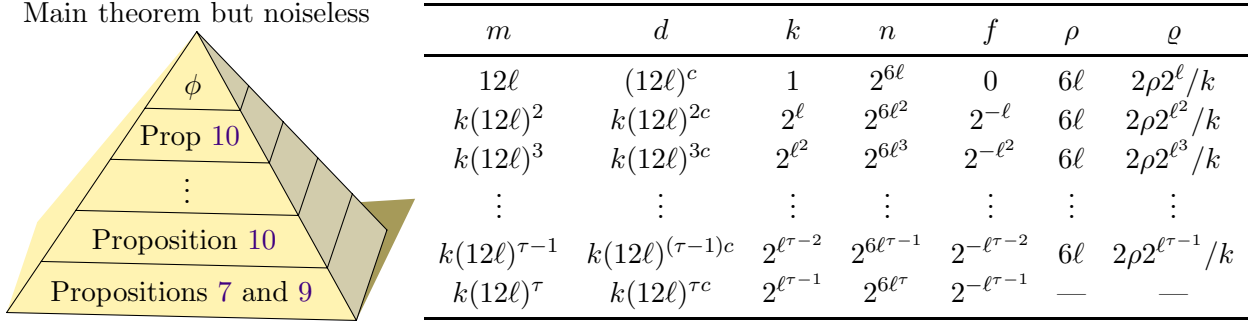

\section{Gacha and Noisy Test Results}                 \label{sec:noise}

    The goal of this section is to address noisy tests.

    Suppose that the channel $Z$ that models the noisy test outputs is a
    BSC.  The itemized analysis in Section~\ref{sec:count} will break
    under $Z$ because it relied on we being able to tell if the Hamming
    weight of $y'_s$ is $0$ (erasure), $3.5\log_2 k$ (uniquely
    decodable), or more (collision).  To fix this, we will employ the
    intuition that good codes are almost--constant-weight codes.  More
    precisely, in a capacity-approaching error-correcting code over a
    BSC, most codewords have Hamming weights $\approx \ell/2$, where
    $\ell$ is the block length.  Moreover, the bitwise-or of two random
    codewords will, with high probability, have Hamming weight $\approx
    3\ell/4$.  Therefore, in order to tell if there are $0$, $1$, or
    more circles in the $s$th row of an array like Figure~\ref{fig:QR},
    it suffices to hypothesis-test if the Hamming weight of $y'_s$ is
    $0$, $\ell/2$, or $3\ell/4$.

    A bigger obstacle is when $Z$ is not a BSC, and for a good reason:
    Most GT applications do not treat positive and negative
    symmetrically.  For instance, bloom filters that disable deletion
    contain FPs (but not FNs).  On the other hands, in digital
    forensics, we can store $1$-bit hashes to save space;
    $1$-bit hashes might collide and produce FNs (but not FPs).
    In Section~\ref{sec:symmetry}, readers will see how to reduce a
    generic channel to a BSC.

\subsection{Good codes have good Hamming properties}    \label{sec:BaZ}

    This subsection is our attempt to fix Section~\ref{sec:count} for
    the case that $Z$ is a BSC.  We borrow Barg and Zémor's codes
    \cite{BaZ02} that satisfy the following properties.
    \begin{itemize}
        \itemsep=0ex
        \item The decoding complexity is linear in $\ell$,
            the block length.
        \item The code rate approaches the capacity of $Z$.
        \item The error exponent is positive at any rate below capacity.
    \end{itemize}
    We want to use these properties to prove the constant-weight
    property of Barg--Zémor.  The constant-distance property follows
    immediately because it is a linear code.  Once we have those, we can
    use Barg--Zémor's encoder as a replacement of the map $\phi$
    defined in Section~\ref{sec:weight}.

    \begin{lemma} [Constant-weight property]          \label{lem:weight}
        Over a BSC $Z$, construct a family of capacity-achieving codes
        using \cite{BaZ02}.  For any $\varepsilon > 0$, there exists a
        $\delta > 0$ and a lower bound on $\ell$ such that the
        probability that a random codeword's Hamming weight is not
        within the range $(1/2 \pm \varepsilon) \ell$ is less than
        $2^{-\delta\ell}$.
    \end{lemma}

    \begin{proof} [Proof of Lemma~\ref{lem:weight}]
        The proof is by a contradiction to Shannon's theory.  In the
        noisy-channel coding theorem, every channel and every input
        distribution has a capacity.  The true capacity of a channel is
        the supremum over all input distributions.  That is, achieving
        the true capacity requires specific input distributions.  For a
        BSC, the only capacity-achieving input distribution is the
        uniform one on $\{0, 1\}$.  Any nonuniform input distribution
        causes a capacity penalty.  We argue that this forces codewords
        to have as many zeros as they have ones.

        Suppose the opposite, that more than $2^{-\delta\ell}$ of the
        codewords have Hamming weights less than $(1/2 - \varepsilon)
        \ell$.  We then communicate using only those codewords.  Doing
        so decreases the effective code rate by $\delta$ but skews the
        input distribution by (at least) $\varepsilon$.  To obtain a
        contradiction, choose $\delta$ so small and $\ell$ so big such
        that the decreased code rate is greater than the penalized
        capacity due to skewed inputs.  This finishes the proof that
        Hamming weight should $> (1/2 - \varepsilon) \ell$.  By a
        mirrored argument, it should $< (1/2 + \varepsilon) \ell$,
        which finishes the proof of Lemma~\ref{lem:weight}.
    \end{proof}

    \begin{lemma} [Constant-distance property]      \label{lem:distance}
        Assume the notation of Lemma~\ref{lem:weight}.  The probability
        that the Hamming weight of the bitwise-or of two random
        codewords is not within the range $(3/4 \pm 3\varepsilon/2)
        \ell$ is less than $3 \cdot 2^{-\delta\ell}$.
    \end{lemma}

    \begin{proof}
        Let $w$ and $w'$ be two random codewords.  Because we are
        dealing with binary linear codes, the Hamming weight of their
        bitwise-or $w \vee w'$ can be expressed as
        \[ |w \vee w'|= \frac{|w|+ |w'|+ |w + w'|}{2}. \]
        Note that $w + w'$ (bitwise-xor) is also a codeword.  Therefore,
        except for a probability of $3 \cdot 2^{-\delta\ell}$, $w$, $w'$
        and $w + w'$ all have Hamming weights within the range $(1/2 \pm
        \varepsilon) \ell$.  This finishes the proof.
    \end{proof}

    Lemmas \ref{lem:weight} and \ref{lem:distance} imply that we can
    tell apart the all-zero codeword, a random codeword, and the
    bitwise-or of two random codewords.

    \begin{proposition} [Positive exponent]         \label{pro:exponent}
        Assume the same notation as in Lemma~\ref{lem:weight}.  Suppose
        that Gill, adversarially, chooses one of the options and sends
        it to Alex via the BSC $Z$: (a) the all-zero codeword, (b) a
        random (non-adversarial) codeword, or (c) the bitwise-or of
        two random (non-adversarial) codewords.  Alex can tell the
        difference with misclassification probability
        $2^{-\gamma \ell}$ for some small $\gamma > 0$.
    \end{proposition}

    \begin{proof}
        Denote the crossover probability of $Z$ by $p$.  Alex knows
        that, if Gill keeps transmitting zeros, he will see a stream of
        bits whose mean is $p$.  If Gill keeps transmitting uniformly
        random bits, Alex will see a stream whose mean is $1/2$.
        Similarly, if Gill's stream has mean $3/4$, Alex will see mean
        $3/4 - p/2$.  Therefore, Alex can use midpoints $\theta_1
        \coloneqq (p + 1/2)/2$ and $\theta_2 \coloneqq (1/2 + 3/4 -
        p/2)/2$ as thresholds to classify incoming streams by sample
        mean.  It remains to compute the probability that Alex makes
        mistakes given this choice of $\theta_1$ and $\theta_2$.

        Suppose that Gill chose (b) and Alex thinks she chose (a).  That
        is, Alex sees $< \theta_1\ell$ ones.  Then either of the
        following must go wrong.
        \begin{itemize}
            \itemsep=0ex
            \item The codeword Gill sent has exceptionally few ones:
                fewer than $3\ell/8$.
            \item The codeword has many ones (more than $3\ell/8$), but
                the outputs of $Z$ have too few (fewer than
                $\theta_1\ell$).
        \end{itemize}
    
        As Lemma~\ref{lem:weight} has shown, a randomly chosen codeword
        has $< 3\ell/8$ ones with probability $2^{-\delta \ell}$ for
        some $\delta > 0$ that is chosen according to $\varepsilon
        \coloneqq 1/8$.  That is to say, the first item is exponentially
        unlikely to go wrong.  At the meantime, the second item is also
        exponentially unlikely to go wrong by Hoeffding's inequality.
        Therefore, we can choose a positive $\gamma$ to upper bound
        the sum of the probabilities that things can go wrong.
        
        The preceding argument focuses on how Alex can distinguish (a)
        and (b).  With a similar argument, Alex can distinguish (b) and
        (c) with an exponentially small misclassification probability.
        This finishes the proof of Proposition~\ref{pro:exponent}.
    \end{proof}

    Summary of this subsection: We argue that Barg--Zémor \cite{BaZ02}
    is a good replacement of the map $\phi$ in Section~\ref{sec:count}.
    We now use it to prove that Gacha can handle test results altered by
    a BSC.

\begin{figure*}
    \centering
    \begin{tikzpicture} [x=3cm, y=0.7cm]
        \draw
            (0, 2) node (1) {$1$} (1, 2) node (l) {$1$}
                                  (1, 1) node (e) {erasure}
            (0, 0) node (0) {$0$} (1, 0) node (o) {$0$}
            (1) edge [->] node [auto] {$1 - p$} (l)
            (1) edge [->] node [auto, '] {$p$} (e)
            (0) edge [->] node [auto] {$p$} (e)
            (0) edge [->] node [auto, '] {$1 - p$} (o)
        ;
    \end{tikzpicture}
    \hfill
    \begin{tikzpicture} [x=3cm, y=0.7cm]
        \draw
            (0, 2) node (1) {$1$} (1, 2) node (l) {$1$}
            (0, 0) node (0) {$0$} (1, 0) node (o) {$0$}
            (1) edge [->] node [auto] {$1$} (l)
            (0) edge [->] node [auto, '] {$q$} (l)
            (0) edge [->] node [auto, '] {$1 - q$} (o)
        ;
    \end{tikzpicture}
    \hfill
    \begin{tikzpicture} [x=3cm, y=0.7cm]
        \draw
            (0, 2) node (1) {$1$} (1, 2) node (l) {$1$}
            (0, 0) node (0) {$0$} (1, 0) node (o) {$0$}
            (1) edge [->] node [auto] {$1 - r$} (l)
            (1) edge [->] node [auto] {$r$} (o)
            (0) edge [->] node [auto, '] {$1$} (o)
        ;
    \end{tikzpicture}
    \hfill
    \begin{tikzpicture} [x=3cm, y=0.7cm]
        \draw
            (0, 2) node (1) {$1$} (1, 2) node (l) {$1$}
            (0, 0) node (0) {$0$} (1, 0) node (o) {$0$}
            (1) edge [->] node [auto] {$1 - s$} (l)
            (1) edge [->] node [auto, ', pos=.2] {$s$} (o)
            (0) edge [->] node [auto, pos=.2] {$s$} (l)
            (0) edge [->] node [auto, '] {$1 - s$} (o)
        ;
    \end{tikzpicture}
    \caption{
        From left to right:
        Binary erasure channel (BEC) with erasure probability $p$.
        False positive (FP) channel with FP probability $q$.
        False negative (FN) channel with FN probability $r$.
        Binary symmetric channel (BSC) with crossover probability $s$.
    }                                                \label{fig:channel}
\end{figure*}
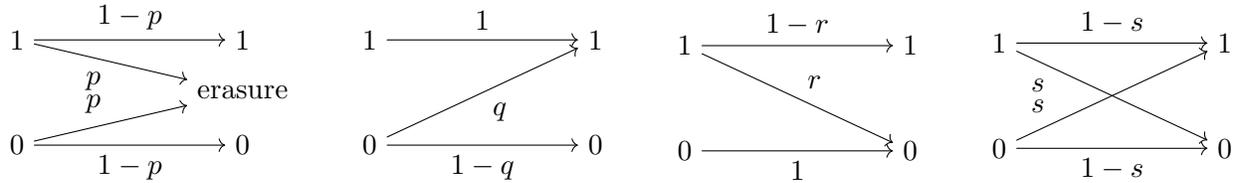

\begin{figure*}
    \centering
    \hbox{}
    \hfill
    \begin{tikzpicture} [x=3cm, baseline=1cm-.5ex]
        \draw
            (0,2) node(1){$1$} (1,2) node(l){$1$} (2,2) node(I){$1$}
            (0,0) node(0){$0$} (1,0) node(o){$0$} (2,0) node(O){$0$}
            (1) edge [->] node [auto] {$1$} (l)
            (0) edge [->] node [auto, '] {$q$} (l)
            (0) edge [->] node [auto, '] {$1 - q$} (o)
            (l) edge [->] node [auto] {$\frac{1}{1 + q}$} (I)
            (l) edge [->] node [auto] {$\frac{q}{1 + q}$} (O)
            (o) edge [->] node [auto, '] {$1$} (O)
        ;
    \end{tikzpicture}
    \hfill
    is equivalent to
    \hfill
    \begin{tikzpicture} [x=3cm, baseline=1cm-.5ex]
        \draw
            (0, 2) node (1) {$1$} (1, 2) node (l) {$1$}
            (0, 0) node (0) {$0$} (1, 0) node (o) {$0$}
            (1) edge [->] node [auto] {$\frac{1}{1+q}$} (l)
            (1) edge [->] node [auto, ', pos=.2, inner sep=0]
            {$\frac{q}{1+q}$} (o)
            (0) edge [->] node [auto, pos=.2, inner sep=0]
            {$\frac{q}{1+q}$} (l)
            (0) edge [->] node [auto, '] {$\frac{1}{1+q}$} (o)
        ;
    \end{tikzpicture}
    \hfill
    \hbox{}
    \caption{
        Take an FP channel with FP probability $q$; one can post-process
        its output by an FN channel with FN probability $q/(1 + q)$ to
        form a BSC with crossover probability $q/(1 + q)$.
    }                                              \label{fig:downgrade}
\end{figure*}
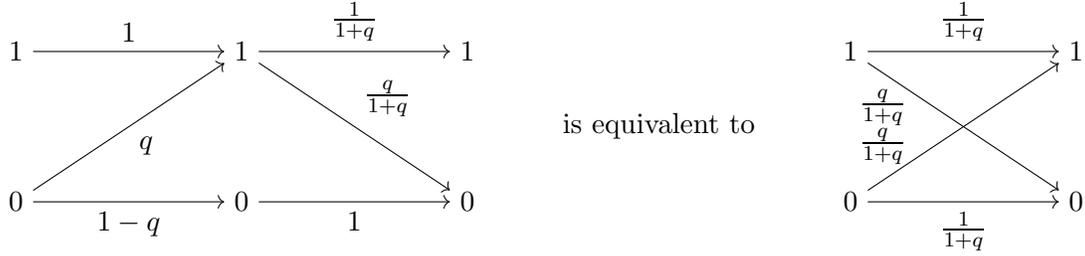

\subsection{Gacha over BSC}
    Recall that $\nu \coloneqq \log_2 n$.  The following theorem has two
    proofs.  One, as shown in Figure~\ref{fig:noisy}, is to regard
    Barg--Zémor as a GT scheme and apply Proposition~\ref{pro:expander}.
    The other proof, the one we present below, is to reprove
    Proposition~\ref{pro:expander} in this new context.

    \begin{theorem} [BSC Gacha]                          \label{thm:BSC}
        If the noisy tests are modeled by a BSC $Z$, then there is a
        constant $\OZ$ and a GT scheme that satisfies the following.
        (A) It uses $m = \OZ(k \nu)$ tests.
        (B) Its decoding complexity is $\OZ(k \poly(\nu))$.
        (C) It will produce $k \exp(-\sqrt\nu)$
        or fewer FPs and FNs on average (over the randomness from $x$,
        $A$, $Z$, and $\DD$).
    \end{theorem}

    \begin{proof}
        The construction is very similar to that in
        Section~\ref{sec:toy}.  The only changes are that $\phi$ is
        replaced by an error-correcting code and that some constants are
        re-chosen.  Here is the list of re-chosen constants.

        \begin{itemize}
            \itemsep=0ex
            \item The size $2^{c_1 \sqrt\nu}$ of the finite field $\FF$
                (cf.\ Section~\ref{sec:fold}).
            \item The degree cap $c_2 \sqrt\nu$ on the polynomials
                $\FF[t]^{<c_2\sqrt\nu}$.  Choose $c_2 = 1/c_1$ so there
                are $n$ polynomials.
            \item The number of evaluation points $c_3 k \sqrt\nu$.
                This is also the number of rows
                in Figure~\ref{fig:circle}.
            \item The number of circles $c_4 \sqrt\nu$ that each column
                in Figure~\ref{fig:circle} has.
            \item The block length $\ell \coloneqq c_5 \sqrt\nu$ of
                $\phi$ (cf.\ Section~\ref{sec:weight}).
        \end{itemize}
        
        Next, to construct $\phi$, use \cite{BaZ02} to construct a code
        whose block length is $\ell = c_5 \sqrt\nu$ and dimension is
        $c_1 \sqrt\nu$ (i.e., $\log_2$ of the field size).  At this code
        rate $c_1/c_5$, denote by $\gamma$ a positive constant such that
        the following hold.
        \begin{itemize}
            \itemsep=0ex
            \item The error exponent is $> \gamma$, that is, the failure
                probability of decoding is $< 2^{-\gamma\ell}$.  By
                \cite{BaZ02}, $\gamma$ exists.
            \item The probability that Alex cannot tell apart all-zero,
                random codeword, and bitwise-or of two random codewords
                is $< 2^{-\gamma\ell}$ (cf.\
                Proposition~\ref{pro:exponent}) .
        \end{itemize}
        In other words, the failure probability
        per batch is about $2^{-\gamma c_5\sqrt\nu}$.

        We now examine the failure probability of the decoder described
        in Section~\ref{sec:count}.  There are three things that could
        go wrong.
        \begin{itemize}
            \itemsep=0ex
            \item Two polynomials $g$ and $h$ coincide at $b_0$.  This
                is addressed in Lemma~\ref{lem:birthday}: As long as the
                field size (controlled by $c_1$) is big enough, the
                collision probability is $\exp(-\OmZ(\sqrt\nu))$.
            \item A polynomial does not contribute sufficiently many
                symbols to the synthesized word or its symbols are
                decoded incorrectly (see the next item).  This is
                addressed in Lemma~\ref{lem:concentrate}: As long as
                $c_3/c_4$ and $c_4/c_2$ are big enough, interpolation
                fails with probability $\exp(-\OmZ(\sqrt\nu))$.
            \item A symbol $(g(b_0), g(p_s))$ is erroneously decoded
                from the binary representation $\phi(g(b_0), g(p_s))$.
                This is addressed when we chose $\gamma$: Because there
                are, on average, $k \exp(-\OmZ(\sqrt\nu))$
                erroneous symbols, they can produce at most $k
                \exp(-\OmZ(\sqrt\nu))$ FPs and FNs.
        \end{itemize}
        From what is presented about, this GT scheme produces $k
        \exp(-\OmZ(\sqrt\nu))$ FNs and FPs.  To conclude (C), apply
        Proposition~\ref{pro:serial} with $\sigma \leftarrow \OZ(1)$ to
        promote the number of mistakes from $k \exp(-\OmZ(\sqrt\nu))$ to
        $k \exp(-\sqrt\nu)$.

        For (A), note that the number of tests is $c_3 k \sqrt\nu$ times
        $c_5 \sqrt\nu$.  For (B), note that $\phi$ has a linear time
        decoder and interpolation is cubic (Gaussian elimination) in the
        number of evaluations.  Together, each polynomial needs
        $\OZ(\sqrt\nu^3)$ time.  There are $k$ polynomials, so the
        overall complexity is $\OZ(k \poly(\nu))$.  This concludes the
        proof of Theorem~\ref{thm:BSC}.
    \end{proof}

    We next address how Gacha should handle non-BSC channels.

\begin{figure*}
    \tikzset{
        every picture/.append style=
            {baseline=0, x={(0.5,-0.05)}, y=0.7cm, z={(0.1,0.1)}},
        side/.style={above, rotate=atan(10)-90}
    }
    \begin{tikzpicture} [3d view={140}{15}, scale=0.6, overlay]
        \draw
            (0, 0) coordinate (A)
            (0, -1, -1) coordinate (B)
            (1, 0, -1) coordinate (C)
            (0, 1, -1) coordinate (D)
            (-1, 0, -1) coordinate (E)
            (-1.4, -0.3, -1) coordinate (E')
        ;
    \end{tikzpicture}
    \begin{tikzpicture} [baseline=($3.5/2*(B)+3.5/2*(D)$)]
        \begin{scope} [transform canvas={scale=3.5}]
            \fill [Gold!30!gray] (B) -- (D) -- (E');
            \fill [Gold!30!gray!50] (A) -- (D) -- (E);
            \fill [Gold!30] (A) -- (B) -- (C) -- (D);
        \end{scope}
        \draw
            (A) node [above, align=center]
                {Theorem~\ref{thm:BSC} \\ BSC Gacha}
            ($2.5*(C)$) -- node [above, align=center]
                {$\phi$ by \\ \cite{BaZ02}}
            ($2.5*(D)$) -- ($2.5*(E)$)
            ($3.5*(C)$) -- node [above] {Proposit~\ref{pro:expander}}
            ($3.5*(D)$) -- ($3.5*(E)$)
            (A) -- ($3.5*(C)$) (A) -- ($3.5*(D)$) (A) -- ($3.5*(E)$)
        ;
    \end{tikzpicture}
    \hfill
    \begin{tikzpicture} [baseline=($4.5/2*(B)+4.5/2*(D)$)]
        \begin{scope} [transform canvas={scale=4.5}]
            \fill [Gold!30!gray] (B) -- (D) -- (E');
            \fill [Gold!30!gray!50] (A) -- (D) -- (E);
            \fill [Gold!30] (A) -- (B) -- (C) -- (D);
        \end{scope}
        \draw
            (A) node [above, align=center]
                {Theorem~\ref{thm:noisy} \\ noisy Gacha}
            ($2.5*(C)$) -- node [above] {Prop~\ref{pro:downgrade}}
            ($2.5*(D)$) -- ($2.5*(E)$)
            ($3.5*(C)$) -- node [above] {$\phi$ by \cite{BaZ02}}
            ($3.5*(D)$) -- ($3.5*(E)$)
            ($4.5*(C)$) -- node [above] {Proposition~\ref{pro:expander}}
            ($4.5*(D)$) -- ($4.5*(E)$)
            (A) -- ($4.5*(C)$) (A) -- ($4.5*(D)$) (A) -- ($4.5*(E)$)
        ;
    \end{tikzpicture}
    \hfill
    \begin{tikzpicture} [baseline=($6.5/2*(B)+6.5/2*(D)$)]
        \begin{scope} [transform canvas={scale=6.5}]
            \fill [Gold!30!gray] (B) -- (D) -- (E');
            \fill [Gold!30!gray!50] (A) -- (D) -- (E);
            \fill [Gold!30] (A) -- (B) -- (C) -- (D);
        \end{scope}
        \draw
            (A) node [above, align=center] {The main theorem}
            ($2.5*(C)$) -- node [above] {Prop~\ref{pro:downgrade}}
            ($2.5*(D)$) -- ($2.5*(E)$)
            ($3.5*(C)$) -- node [above] {$\phi$ by \cite{BaZ02}}
            ($3.5*(D)$) -- ($3.5*(E)$)
            ($5.5*(C)$) -- node [above, align=center] {
                Proposition~\ref{pro:expander}\\[-7pt]$\vdots$\\[-4pt]
                Proposition~\ref{pro:expander}
            }
            ($5.5*(D)$) -- ($5.5*(E)$)
            ($6.5*(C)$) -- node [above]
            {Propositions \ref{pro:parallel} and \ref{pro:serial}}
            ($6.5*(D)$) -- ($6.5*(E)$)
            (A) -- ($6.5*(C)$) (A) -- ($6.5*(D)$) (A) -- ($6.5*(E)$)
        ;
    \end{tikzpicture}
    \caption{
        Stacking $\tau - 1$ layers of Proposition~\ref{pro:expander}
        together with other gadgets that denoise the test results
        proves the main theorem.
    }                                                  \label{fig:noisy}
\end{figure*}
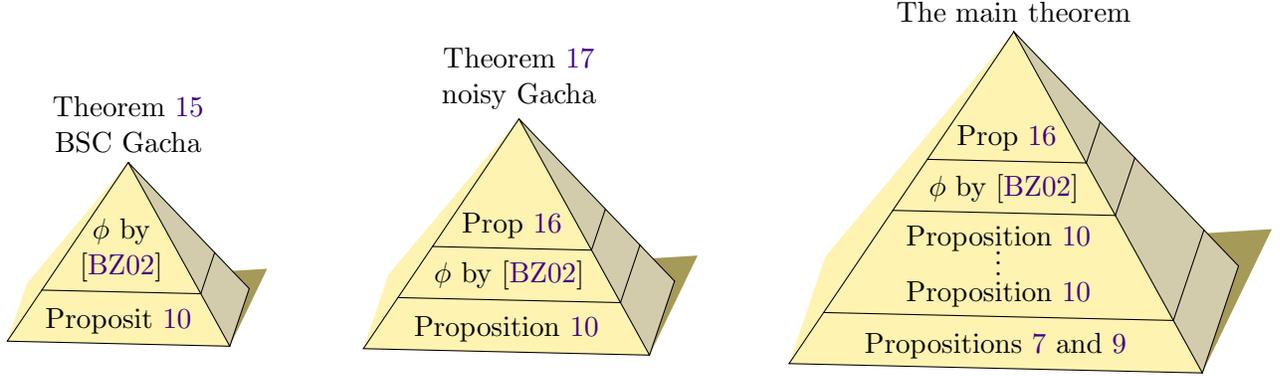

\subsection{Symmetrizing asymmetric channels}       \label{sec:symmetry}

    As explained at the beginning of this section, asymmetric channels
    are natural sources of noises in GT applications.  In this
    subsection, we demonstrate how to downgrade an asymmetric channel to
    a BSC.  As an example, see Figure~\ref{fig:downgrade} for how to
    downgrade an FP channel into a BSC.

    \begin{proposition} [Downgrade to BSC]         \label{pro:downgrade}
        Any binary-input discrete-output channel with a positive
        capacity can be downgraded into a BSC with a positive capacity.
    \end{proposition}

    \begin{proof}
        Suppose that the channel is $Z = (\Sigma, \mu_0, \mu_1)$ where
        $\mu_x(y)$, for $(x, y) \in \{0, 1\} \times \Sigma$, is the
        probability that the output is $y$ conditioning on the input
        being $x$.  Suppose that the alphabet $\Sigma = \{\sigma_1,
        \sigma_2, \dotsc, \sigma_q\}$ is sorted by \emph{likelihood
        ratios}:
        \[
            \frac{\mu_1(\sigma_1)}{\mu_0(\sigma_1)}
            \leq \frac{\mu_1(\sigma_2)}{\mu_0(\sigma_2)}
            \leq \dotsb
            \leq \frac{\mu_1(\sigma_q)}{\mu_0(\sigma_q)}.
        \]
        Define a family of discrete-input binary-output channels
        $V_t\colon \Sigma \to \{0, 1\}$ that is parametrized by $t \in
        [0, q]$ as follows: For all $i \in [q]$, $V_t$ sends $\sigma_i$
        to $0$ if $i < t + U$, and to $1$ if $i > t + U$, where $U \in
        [0, 1]$ is a uniform random seed.

        When $t = 0$, $V_0$ sends all symbols in $\Sigma$ to $1$.  When
        $t = q$, $V_q$ sends everything to $0$.  As $t$ gradually
        increases from $0$ to $q$, $V_t$ sends more and more symbols
        (starting from those that are more likely $0$) to $0$ and fewer
        symbols to $1$.  We define a channel $W_t$ by post-processing
        $Z$ by $V_t$.  According to what we have above, $W_0$ is an FP
        channel, $W_q$ is an FN channel, and $W_t\colon \{0, 1\} \to
        \{0, 1\}$ is a generic binary-to-binary channel.  By the
        intermediate value theorem, there is an intermediate $t \in [0,
        q]$ such that $W_t$ is a BSC.
    \end{proof}
    
    Proposition~\ref{pro:downgrade} suggests that, if the test outputs
    are processed by a non-BSC channel $Z$, we can post-process the test
    outputs further using $W_t$ (for some proper $t \in [0, q]$).  Doing
    so makes the GT problem equivalent to what is addressed in
    Theorem~\ref{thm:BSC}.  This generalizes Theorem~\ref{thm:BSC}
    to generic channels.

    \begin{theorem} [Noisy Gacha]                      \label{thm:noisy}
        If tests are noisy and modeled by a binary-input channel $Z$,
        then there is a GT scheme that satisfies the following.
        (A) It uses $m = \OZ(k \nu)$ tests.
        (B) Its decoding complexity is $\OZ(k \poly(\nu))$.
        (C) It will produce $k \exp(-\sqrt\nu)$ or fewer FPs and FNs
        on average (over the randomness from $x$, $A$, $Z$, and $\DD$).
    \end{theorem}

\section{Wrap up the Proof of the Main Theorem}         \label{sec:main}

    We have seen how to stack gadgets.
    We have seen how to denoise tests.
    Putting everything together proves the main theorem.

    \bgroup
    \def\thetheorem{\addtocounter{theorem}{-1}}
    \begin{theorem} [Restating Theorem~\ref{thm:main} on
                     page~\pageref{thm:main} for readers' convenience]
        Let $Z$ be any binary-input channel and $\OZ$ hide a constant
        that depends only on $Z$.  Let $\sigma \geq 1$ and $\tau \geq 2$
        be free integer parameters.  Gacha GT is a randomized scheme
        that uses $m = \OZ(\sigma k \log(n) 2^{\OO(\tau)})$ tests and
        decoding complexity $\OZ(\sigma k \poly(\log(\sigma n))
        2^{\OO(\tau)})$ to find $k$ sick people in a population of $n$.
        Averaged over the randomness of $x$, $A$, and $Z$, Gacha will
        produce $k \exp(-\sigma \log_2(n)^{1-1/\tau})$ or fewer FPs
        and FNs.
    \end{theorem}
    \egroup

    \begin{proof} [Proof of Theorem~\ref{thm:main}]
        The proof is visualized in Figure~\ref{fig:noisy}.

        In Proposition~\ref{pro:downgrade}, we learned how to downgrade
        any binary-input channels to a BSC; this would be the tip of the
        pyramid as this applies to test results before everything else.
        (Recall that tests are at the tip and people are at the base.)
        We learned how to handle BSCs in Theorem~\ref{thm:BSC}: use an
        error-correcting code as $\phi$; this is the second layer of the
        pyramid.

        We next want to perform mathematical induction on $\tau$.
        Iteratively we apply Proposition~\ref{pro:expander} with
        lowercase parameters
        \begin{center}
            \begin{tabular}{ccccccc}
                \toprule
                $m$ & $d$ & $k$ & $n$ & $f$ & $\rho$ & $\varrho$
                \\ \midrule
                $\OZ(k (12\ell)^{\tau-1})$
                & $\OZ(k (12\ell)^{(\tau-1)c})$ & $2^{\ell^{\tau-2}}$
                & $2^{6\ell^{\tau-1}}$ & $\exp(-\OmZ(\ell^{\tau-2}))$
                & $6\ell$ & $2\rho 2^{\ell^{\tau-1}} / k$
                \\ \bottomrule
            \end{tabular}
        \end{center}
        to obtain capital parameters
        \begin{center}
            \begin{tabular}{ccccc}
                \toprule
                $M$ & $D$ & $K$ & $N$ & $F$
                \\ \midrule
                $\varrho m$
                & $\varrho d$ + $\varrho k \poly(\rho \log(\varrho n))$
                & $\varrho k / 2\rho$ & $n^{\rho/6}$
                & $\substack{
                    2^\rho (f/k + 2^{-\Omega(k)})^{\rho/3} K
                    + K^2 / 2 \sqrt n
                    \\ + \varrho^{\rho/3}
                    (f + 2^{-\Omega(k)})^{\rho/3} \sqrt n^{1-\rho/3}
                }$
                \\ $\OZ(2^{\ell^{\tau-1}} (12\ell)^\tau)$
                & $\OZ(2^{\ell^{\tau-1}} (12\ell)^{\tau c})$
                & $2^{\ell^{\tau-1}}$ & $2^{6\ell^\tau}$
                & $\exp(-\OmZ(\ell^{\tau-1}))$
                \\ \bottomrule
            \end{tabular}
        \end{center}
        This proves that Gacha can find
        $k$ sick people in a population of
        $n$, for when $\log k \approx \log(n)^{1-1/\tau}$, using
        $\OZ(k \log(n) 2^{\OO(\tau)})$ tests and complexity
        $\OZ(k \poly(\log n) 2^{\OO(\tau)})$, and, on average, produces
        $k \exp(-\OmZ(\log_2(n)^{1-1/\tau}))$ FPs and FNs.

        We apply Proposition~\ref{pro:parallel} to relax the restriction
        $\log k \approx \log(n)^{1-1/\tau}$ to allow arbitrary $k$.  We
        also apply Proposition~\ref{pro:serial} to cover the case
        $\sigma \geq 2$.  These two propositions constitute the bottom
        layer of the pyramid and conclude the proof.
    \end{proof}

\clearpage

\appendix
\def\thesection{Appendix~\Alph{section}}

\section{Old Group Testing Techniques}                \label{app:tricks}

    Recall $\nu \coloneqq \log_2 n$.  Recall that every person
    has a \emph{phone number} in $\{0, 1\}^\nu$.
    The goal of GT is to recover the phone numbers of all the sick
    people.  We introduce GT techniques used by earlier works in this
    appendix.  The goal is to show readers how some of them could help
    in our situation or why the others could not.

\subsection{COMP and DD}

    COMP stands for combinatorial orthogonal matching pursuit.  It is a
    decoding algorithm that can work with arbitrary configuration
    matrices, but the most effective ones are random.  COMP is based on
    the following rule of thumb: If a person participates in a test that
    turns out to be negative, then we can infer that she is definitely
    healthy; otherwise we are not sure, so we declare her sick.

    DD stands for definite defective.  DD is a decoding algorithm that
    can work with arbitrary configuration matrices, but the most
    effective ones are random.  DD is based on the following rule of
    thumb: If a person participates in a test that turns out to be
    positive, and every other person participating in this test is
    definitely healthy in sense of COMP, then we can infer that she is
    definitely sick; otherwise we are not sure, so we declare her
    healthy.

    To summarize, COMP and DD are greedy algorithms, one biased toward
    sick and the other toward healthy.  Both algorithms fall apart when
    tests are noisy, as one negative test does not imply wellness of the
    participants.  Therefore, in noisy COMP, we declare a person sick if
    she participates in more than $\theta$ tests, where $\theta$ is some
    carefully chosen threshold.  This three algorithms all suffer from a
    $\Omega(n)$ penalty on decoding complexity.

\subsection{List-disjunct matrices}

    Unlike $k$-disjunct matrices, which are configuration matrices able
    to identify $k$ sick people, a $(k, \ell)$--list-disjunct matrix is
    a configuration matrix that identifies a list of $\ell$ suspicious
    people that contain the sick $k$.  One then relies on other
    mechanisms to reduce the list of $\ell$ to $k$, often by other
    list-disjunct matrices.  See Figure~\ref{fig:cheese} for a cartoon
    of that.

    An iconic usage of list-disjunct matrices is Kautz and Singleton's
    construction \cite{KaS64} accelerated by Inan, Kairouz, Wootters,
    and Özgur \cite{IKW19}.  While Kautz--Singleton is just one-hot
    encoding applied to Reed--Solomon codes, the acceleration part is to
    take a list-disjunct matrix $A \in \{0, 1\}^{m\times\sqrt n}$ for
    population $\sqrt n$ and construct
    \[
        \AA \coloneqq
        \bma{
            \kern1.9em A \otimes 1^{1\times\sqrt n}
            \\ 1^{1\times\sqrt n} \otimes A \kern1.9em
            \\ \text{COMP-like}
        }
        \in \{0, 1\}^{(2m+?)\times n}
    \]
    In plain text, \cite{IKW19} imagines that a phone number is a pair
    of digits in $[\sqrt n]^2$, uses $A \otimes 1^{1\times\sqrt n}$ to
    obtain a list of the left digits, uses $1^{1\times\sqrt n} \otimes
    A$ to obtain a list of the right digits, and uses a COMP-like design
    at the bottom of $\AA$ to match the left and right digits.

    The matching step incurs a $\Omega(\ell^2)$ complexity penalty as it
    enumerates over all possible combination of left and right digits.
    In contrast, our complexity target is $k \poly(\nu)$, which forbids
    us from any sort of enumeration.

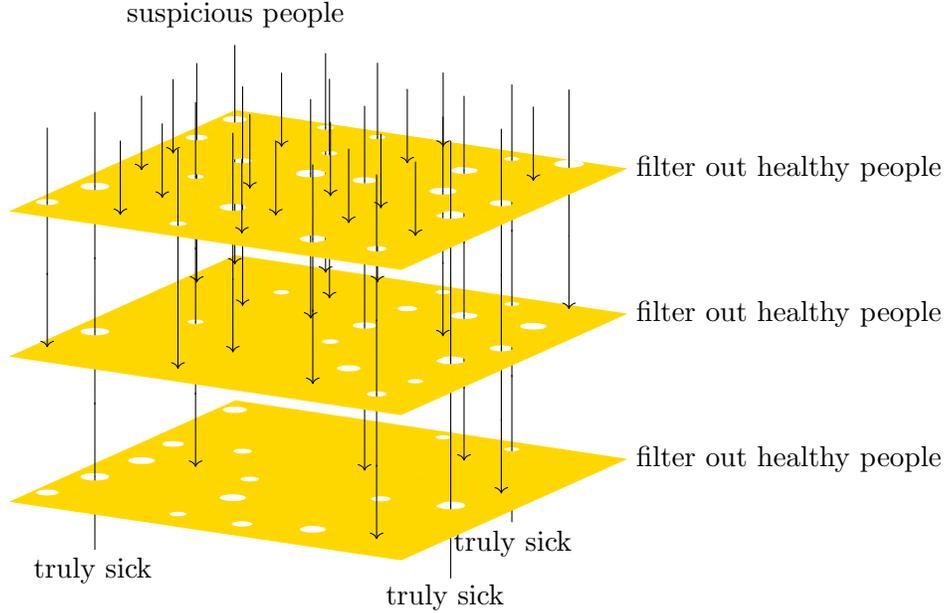
\begin{figure*}
    \centering
    \begin{tikzpicture} [3d view={120}{15}]
        \def\1{1.02}
        \draw (0, 0, 7) node [above] {suspicious people};
        \def\z{0}
        \pgfmathsetseed{8881616}
        \foreach \x in {0, ..., 5} {
            \foreach \y in {0, ..., 5} {
                \pgfmathsetmacro\W{0.1+random/15}
                \pgfmathsetmacro\X{\x.5 - random/2}
                \pgfmathsetmacro\Y{\y.5 - random/2}
                \pgfmathsetmacro\Z{\z*2}
                \pgfmathtruncatemacro\r{random(0, 14)}
                \pgfmathtruncatemacro\h{mod(floor(\r/2^\z), 2)}
                \pgfmathsetmacro\a{\r < 2^(\z+1)}
                \ifnum \r < 2
                    \draw (\X, \Y, 1) node [below] {truly sick};
                \fi
            }
        }
        \foreach \z in {1, 2, 3} {
            \draw (0, 6, \z*2) node [right] {filter out healthy people};
            \pgfmathsetseed{8881616}
            \foreach \x in {0, ..., 5} {
                \foreach \y in {0, ..., 5} {
                    \pgfmathsetmacro\W{0.1+random/10}
                    \pgfmathsetmacro\X{\x.5 - random/3}
                    \pgfmathsetmacro\Y{\y.5 - random/3}
                    \pgfmathsetmacro\Z{\z*2}
                    \pgfmathtruncatemacro\r{random(0, 14)}
                    \pgfmathtruncatemacro\h{mod(floor(\r/2^\z), 2)}
                    \pgfmathsetmacro\a{\r < 2^(\z+1)}
                    \ifnum \h = 0
                        \ifnum \a = 1
                            \draw (\X, \Y, \Z-.1) -- (\X, \Y, \Z-1);
                            \fill [Gold, even odd rule]
                                (\x, \y, \Z) -- (\x, \y+\1, \Z) --
                                (\x+\1, \y+\1, \Z) -- (\x+\1, \y, \Z)
                                (\X, \Y, \Z) circle (\W)
                            ;
                            \draw (\X, \Y, \Z+\1) -- (\X, \Y, \Z);
                            \scoped{
                                \clip (\X, \Y, \Z) circle (\W);
                                \draw (\X, \Y, \Z) -- (\X, \Y, \Z-1);
                            }
                        \else
                            \fill [Gold, even odd rule]
                                (\x, \y, \Z) -- (\x, \y+\1, \Z) --
                                (\x+\1, \y+\1, \Z) -- (\x+\1, \y, \Z)
                                (\X, \Y, \Z) circle (\W)
                            ;
                        \fi
                    \else
                        \ifnum \a = 1
                            \fill [Gold]
                                (\x, \y, \Z) -- (\x, \y+\1, \Z) --
                                (\x+\1, \y+\1, \Z) -- (\x+\1, \y, \Z)
                            ;
                            \draw [->] 
                                (\X, \Y, \Z+\1) -- (\X, \Y, \Z);
                        \else
                            \fill [Gold]
                                (\x, \y, \Z) -- (\x, \y+\1, \Z) --
                                (\x+\1, \y+\1, \Z) -- (\x+\1, \y, \Z)
                            ;
                        \fi
                    \fi
                }
            }
        }
    \end{tikzpicture}
    \caption{
        List-disjunct matrices are like Swiss cheese.  A list-disjunct
        matrix has more than $k$ holes so it cannot tell exactly who are
        sick.  But stacking properly designed matrices on top of each
        other will reveal the true $k$ holes that correspond to the
        truly sick people.
    }                                                 \label{fig:cheese}
\end{figure*}

\subsection{Nonadaptive splitting}

    The word \emph{splitting} refers to an old technique studied by
    Hwang \cite{Hwa72} and earlier authors (see the references therein),
    which boils down to binary search with some carefully crafted
    parameters.  Consuming the optimal number of tests up to an additive
    $\OO(k)$, these splitting schemes are, nevertheless, adaptive,
    meaning that they look at earlier test results to optimize how to
    cut the test takers into groups.

    Nonadaptive splitting aims to eliminate the adaptive part by
    performing randomized tests all at once and hoping that a
    considerable fraction of these test will become useful.  It was
    developed concurrently in \cite{ChN20} and \cite{PrS20} with number
    of test a big-O of $k\nu$.  This constant was later improved in
    \cite{BonsaiGT} to match the constant of COMP.

    To facilitate a decoder faster than COMP, the core design of
    nonadaptive splitting is as follows: Tests are divided into $\nu$
    batches, each consisting of $O(k)$ tests.  For the $b$th batch,
    we hash the first $b$ bits of a phone number to determine which test
    she should be in.  This design is pictured in
    Figure~\ref{fig:split}.  To decode, travel from the roots to every
    leaf that is not blocked by any negative test; the reachable leafs
    correspond to people who are very likely sick, as is visualized in
    Figure~\ref{fig:split}.

    This idea does not work when tests are noisy because, similar to
    before, one negative test does not imply wellness.  The work of
    Price, Scarlett, and Tan \cite{PST23} attempts to fix this by
    ``looking ahead'' the branches to collect more evidences.  This,
    however, results in a decoding complexity inferior to Gacha.

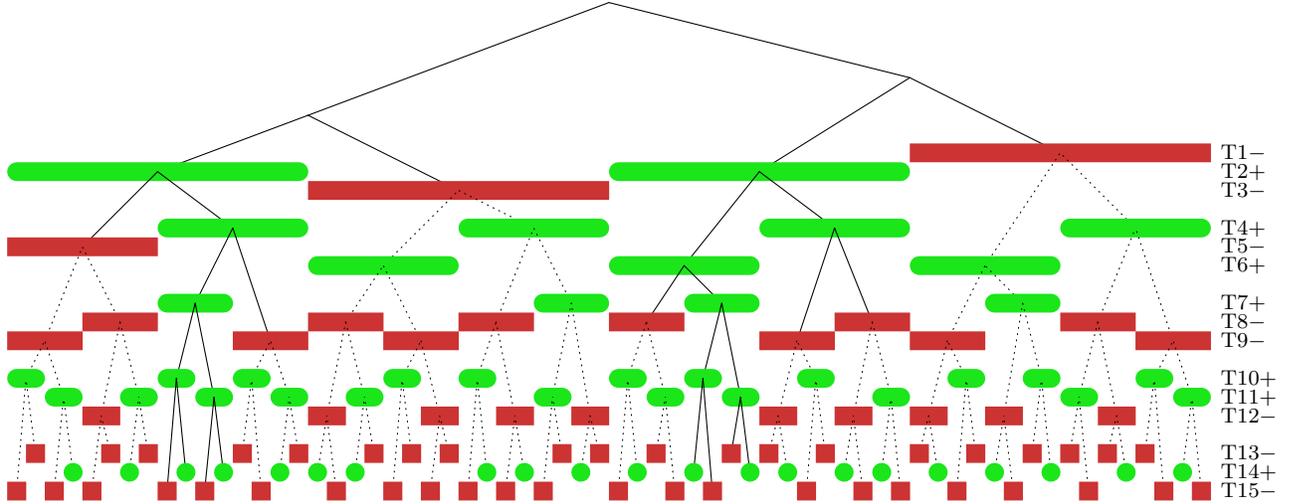
\begin{figure*}
    \centering
    \makeatletter
    \def\pgfkeyssetxvalue#1#2{%
        \expandafter\xdef\csname pgfk@#1\endcsname{#2}%
    }
    \pgfkeys{/handlers/.xinitial/.code={
        \pgfkeyssetxvalue{\pgfkeyscurrentpath}{#1}
    }}
    \pgfmathsetseed{8881616}
    \xdef\h{0}
    \foreach \k in {0, ..., 128}{
        \pgfmathtruncatemacro\h{mod(\h + random(2), 3)}
        \xdef\h{\h}
        \pgfkeys{/hash/\k/.xinitial={\h}}
    }
    \pgfmathdeclarefunction{Hash}{1}{%
        \pgfmathtruncatemacro\k{#1}%
        \pgfkeys{/hash/\k/.get=\pgfmathresult}%
    }
    \makeatother
    \tikzset{
        positive/.style={green!80!gray, line cap=round},
        negative/.style={red!60!gray, line cap=rect}
    }
    \begin{tikzpicture} [x=1cm, y=1cm]
        \def\sickA{11}
        \def\sickB{36}
        \foreach \L in {0, ..., 6}{
            \pgfmathsetmacro\twotoL{2^\L}
            \pgfmathsetmacro\twotoLminusone{\twotoL - 1}
            \pgfmathtruncatemacro\posiA{Hash(\twotoL * (1 + \sickA/64))}
            \pgfmathtruncatemacro\posiB{Hash(\twotoL * (1 + \sickB/64))}
            \pgfkeys{/outcome/0/.initial={-}}
            \pgfkeys{/outcome/1/.initial={-}}
            \pgfkeys{/outcome/2/.initial={-}}
            \pgfkeys{/outcome/\posiA/.initial={+}}
            \pgfkeys{/outcome/\posiB/.initial={+}}
            \ifnum \L > 1
                \foreach \Tmod in {0, 1, 2}{
                    \pgfkeys{/outcome/\Tmod/.get=\outcome}
                    \pgfmathtruncatemacro\Tid{\L*3 + \Tmod - 5}
                    \node at (16, -\L - \Tmod/4) [right, overlay]
                    {\scriptsize T\Tid$\outcome$};
                }
            \fi
            \foreach \P in {0, ..., \twotoLminusone}{
                \pgfmathtruncatemacro\Tmod{Hash(\twotoL + \P)}
                \pgfkeys{/outcome/\Tmod/.get=\outcome}
                \draw ({(16*\P + 8)/\twotoL}, -\L - \Tmod/4)
                    coordinate(L\L/P\P);
                \ifnum \L > 0
                    \pgfmathtruncatemacro\LL{\L - 1}
                    \pgfmathtruncatemacro\PP{floor(\P/2)}
                    \pgfkeys{/innocent/\LL/\PP/.get=\innocent}
                    \ifx \innocent \relax
                        \draw (L\L/P\P) -- (L\LL/P\PP);
                    \else
                        \draw [dotted] (L\L/P\P) -- (L\LL/P\PP);
                        \pgfkeys{/innocent/\L/\P/.xinitial={innocent}}
                    \fi
                \fi
                \ifnum \L > 1
                    \if \outcome -
                        \tikzset{strip/.style=negative}
                        \pgfkeys{/innocent/\L/\P/.xinitial={innocent}}
                    \else
                        \tikzset{strip/.style=positive}
                    \fi
                    \draw [line width=1cm/4, strip]
                        ({(16*\P + 0)/\twotoL + 1/8}, -\L - \Tmod/4)
                        -- ({(16*\P + 16)/\twotoL - 1/8}, -\L - \Tmod/4)
                        -- +(0.1pt, 0)
                    ;
                \fi
            }
        }
    \end{tikzpicture}
    \caption{
        Nonadaptive splitting \cite{ChN20, PrS20, BonsaiGT}.  The
        configuration matrix (strips) and the decoding process (tree).
        Each person is a column $1/4$ centimeters ($1/10$ inches) wide. 
        T1--T15 are tests.  The colored strips are places where the
        encoding matrix has $1$.  Green round strips symbolize positive
        tests; red rectangular strips symbolizes negative tests.  A
        person is declared infected if it is all green light when
        traveling from the root to the leaf. See also
        \cite[Figure~1]{PrS20} and \cite[Fig.~1]{PST23}.
    }                                                  \label{fig:split}
\end{figure*}

\subsection{GROTESQUE and SAFFRON: writing via expander graphs}

    GROTESQUE \cite{CJB17} stands for group testing quick and efficient.
    It divides tests into $O(\sigma k)$ batches, each containing $\nu$
    tests.  Each person then chooses $O(\sigma)$ batches; in each chosen
    batch, she will participate in tests at the locations where her
    phone number has $1$.  We say that she copy-and-pastes her phone
    number to the batches.

    To decode GROTESQUE, we need a mechanism to determine if a batch
    contains zero, one, or more sick people.  Suppose that we have an
    oracle that does that, then from a batch that contains exactly one
    sick person we can recover her phone number.

    In the same vein, SAFFRON \cite{LCP19} works with an oracle that
    determines if a batch contains zero, one, two, or more sick people.
    From the batches that contain exactly one sick person SAFFRON can
    recover her phone number.  From the batches that contain exactly two
    sick people, granted that one of them is already recovered, SAFFRON
    can recover the other phone number.

    A sick person cannot be identified if every of her batches contains
    two (for GROTESQUE) or three (for SAFFRON) sick people.  This leads
    to the heuristic that people should choose batches in the manner of
    expander graphs.  But there is still one fundamental issue: When
    $\sigma$ is constant, a constant fraction of sick people will become
    invisible due to their batches being overflown by too many sick
    people.  To make sick people more discoverable we need to increase
    $\sigma$, which demands more tests and drives us away from the
    optimal $m$.  We resolve this issue by letting everyone choose more
    batches but each batch only contains a fraction of her phone number,
    i.e., we codify the phone numbers.  This way, they have their
    discoverability increased but the number of tests remains the same,
    as can be seen by contrasting Figure~\ref{fig:horizontal} to
    Figure~\ref{fig:QR}.

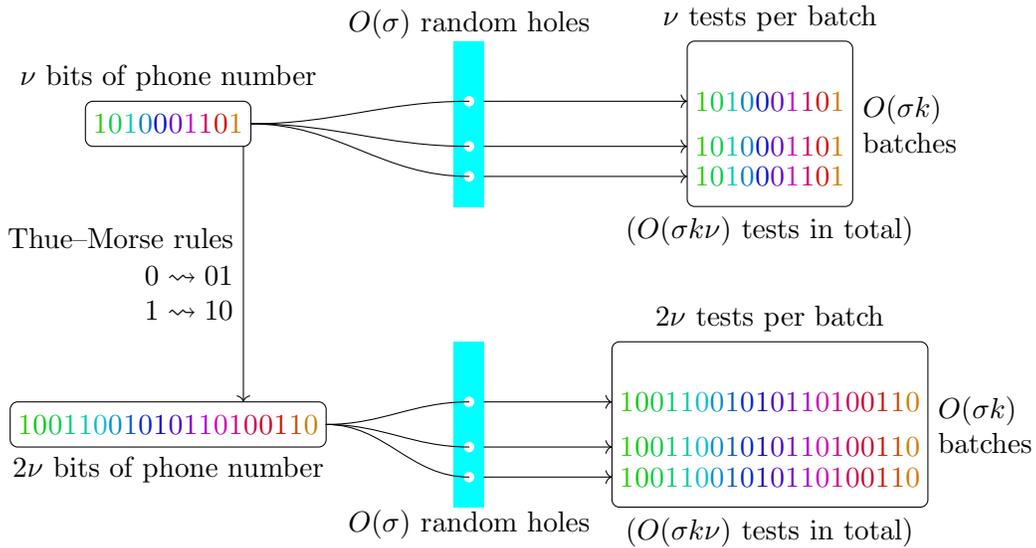
\begin{figure*}
    \centering
    \begin{tikzpicture}
        \resetcolorseries[9]{noyellow}
        \begin{scope} [shift={(-4, 2)}]
            \pgfmathsetseed{8881616}
            \foreach \x in {0, ..., 9}{
                \pgfmathtruncatemacro\b{random(0, 1)}
                \draw [{noyellow!![\x]}] (\x/5-9/10, 0) node {$\b$};
            }
            \draw [rounded corners=0.1cm]
                (-1.1, -0.3) rectangle (1.1, 0.3)
                (0, 0.3) node [above] {$\nu$ bits of phone number}
                (1.1, 0) coordinate (left);
            ;
        \end{scope}
        \begin{scope} [shift={(0, 2)}]
            \fill [cyan]
                (-0.2, 1.1) rectangle (0.2, -1.1)
                (0, 1) node [above, black] {$O(\sigma)$ random holes}
                (0, 0) coordinate (middle)
            ;
            \foreach \y in {1, 3, 6} {
                \fill [white] (0, \y/5-9/10) circle (2pt);
            }
        \end{scope}
        \begin{scope} [shift={(4, 2)}]
            \foreach \y in {1, 3, 6}{
                \pgfmathsetseed{8881616}
                \foreach \x in {0, ..., 9}{
                    \pgfmathtruncatemacro\b{random(0, 1)}
                    \draw [{noyellow!![\x]}]
                        (\x/5-9/10, \y/5-9/10) node {$\b$};
                }
                \draw [->]
                    (middle) +(0, \y/5-9/10) coordinate(mid)
                    (left) to [out=0, in=180] (mid)
                    +(0.2, 0) -- (-1.1, \y/5-9/10)
                ;
            }
            \draw [rounded corners=0.1cm]
                (-1.1, -1.1) rectangle (1.1, 1.1)
                (0, 1.1) node [above] {$\nu$ tests per batch}
                (1.1, 0)
                node [right, align=left] {$O(\sigma k)$ \\ batches}
                (0, -1.1)
                node [below] {($O(\sigma k \nu)$ tests in total)}
            ;
        \end{scope}
        \draw [->]
            (-3, 1.7) -- node [left, align=right] {
                Thue--Morse rules
                \\ $0 \leadsto 01$
                \\ $1 \leadsto 10$
            }
            (-3, -1.7)
        ;
        \resetcolorseries[19]{noyellow}
        \begin{scope} [shift={(-4, -2)}]
            \pgfmathsetseed{8881616}
            \foreach \x in {0, 2, ..., 18}{
                \pgfmathtruncatemacro\b{random(0, 1)}
                \draw [{noyellow!![\x]}] (\x/5-19/10, 0) node {$\b$};
                \pgfmathtruncatemacro\b{1 - \b}
                \pgfmathtruncatemacro\x{\x + 1}
                \draw [{noyellow!![\x]}] (\x/5-19/10, 0) node {$\b$};
            }
            \draw [rounded corners=0.1cm]
                (-2.1, -0.3) rectangle (2.1, 0.3)
                (0, -0.3) node [below] {$2\nu$ bits of phone number}
                (2.1, 0) coordinate (left);
            ;
        \end{scope}
        \begin{scope} [shift={(0, -2)}]
            \fill [cyan]
                (-0.2, 1.1) rectangle (0.2, -1.1)
                (0, -1) node [below, black] {$O(\sigma)$ random holes}
                (0, 0) coordinate (middle)
            ;
            \foreach \y in {1, 3, 6} {
                \fill [white] (0, \y/5-9/10) circle (2pt);
            }
        \end{scope}
        \begin{scope} [shift={(4, -2)}]
            \foreach \y in {1, 3, 6}{
                \pgfmathsetseed{8881616}
                \foreach \x in {0, 2, ..., 18}{
                    \pgfmathtruncatemacro\b{random(0, 1)}
                    \draw [{noyellow!![\x]}]
                        (\x/5-19/10, \y/5-9/10) node {$\b$};
                    \pgfmathtruncatemacro\b{1 - \b}
                    \pgfmathtruncatemacro\x{\x + 1}
                    \draw [{noyellow!![\x]}]
                        (\x/5-19/10, \y/5-9/10) node {$\b$};
                }
                \draw [->]
                    (middle) +(0, \y/5-9/10) coordinate(mid)
                    (left) to [out=0, in=180] (mid)
                    +(0.2, 0) -- (-2.1, \y/5-9/10)
                ;
            }
            \draw [rounded corners=0.1cm]
                (-2.1, -1.1) rectangle (2.1, 1.1)
                (0, 1.1) node [above] {$2\nu$ tests per batch}
                (2.1, 0)
                node [right, align=left] {$O(\sigma k)$ \\ batches}
                (0, -1.1)
                node [below] {($O(\sigma k \nu)$ tests in total)}
            ;
        \end{scope}
    \end{tikzpicture}
    \caption{
        Upper half: GROTESQUE \cite{CJB17}.  Copy-and-paste the phone
        number to random batches of tests.  Lower half: SAFFRON
        \cite{LCP19}.  Apply Thue--Morse rules before copy-any-pasting.
        Usually $\sigma \leftarrow \log_2 k$ for exact recovery.  It is
        also possible to use $\sigma \leftarrow 1$ to achieve an
        order-optimal number of tests at the cost of $\varepsilon k$ FNs
        and FPs.
    }                                             \label{fig:horizontal}
\end{figure*}

\begin{figure*}
    \centering
    \begin{tikzpicture}
        \resetcolorseries[14]{noyellow}
        \begin{scope} [shift={(-4, 2)}]
            \pgfmathsetseed{8881616}
            \foreach \x in {0, ..., 9}{
                \pgfmathtruncatemacro\b{random(0, 1)}
                \draw [{noyellow!![\x]}] (\x/5-9/10, 0) node {$\b$};
            }
            \draw [rounded corners=0.1cm]
                (-1.1, -0.3) rectangle (1.1, 0.3)
                (0, 0.3) node [above] {$\nu$ bits of phone number};
            ;
            \draw [->] (0, -0.3) -- ++ (0, -0.2)
                node [below left, overlay] {error-correcting code}
                arc (180:270:1) -- +(0.7, 0) 
            ;
        \end{scope}
        \begin{scope} [shift={(-2, -1)}]
            \pgfmathsetseed{8883232}
            \foreach \x in {0, ..., 14}{
                \pgfmathtruncatemacro\b{random(0, 1)}
                \draw [{noyellow!![\x]}]
                    (0, \x/5-14/10) node [rotate=90] {$\b$};
            }
            \draw [rounded corners=0.1cm]
                (-0.3, -1.6) rectangle (0.3, 1.6)
                (-0.3, 0) node [left, align=right]
                {$O(\nu)$ bits of \\ encoded \\ phone \\ number}
                (0.3, 0) coordinate (left)
            ;
        \end{scope}
        \begin{scope}
            \draw [double]
                (-0.2, 3.1) to [out=120, in=60] (0.2, 3.1)
                (-0.2, -3.1) to [out=240, in=300] (0.2, -3.1)
            ;
            \fill [cyan]
                (-0.2, 3.1) rectangle (0.2, -3.1)
                (0, -3.1) node [below, black] {mask}
                (0, 0) coordinate (middle)
            ;
            \foreach \Y in
            {0,1,3,4,5,9,10,11,13,18,19,24,27,28,29} {
                \fill [white] (0, \Y/5-29/10) circle (2pt);
            }
        \end{scope}
        \begin{scope} [shift={(2, 0)}]
            \pgfmathsetseed{8883232}
            \foreach [count=\c] \Y in
            {0,1,3,4,5,9,10,11,13,18,19,24,27,28,29} {
                \pgfmathtruncatemacro\b{random(0, 1)}
                \pgfmathtruncatemacro\y{\c - 1}
                \draw [{noyellow!![\y]}]
                    (0, \Y/5-29/10) node [rotate=90] {$\b$};
                \draw [->]
                    (middle) +(0, \Y/5-29/10) coordinate (temp)
                    (left) +(0, \y/5-14/10) to [out=0, in=180] (temp)
                    +(0.2, 0) -- (-0.3, \Y/5-29/10)
                ;
            }
            \draw [rounded corners=0.1cm]
                (-0.3, -3.1) rectangle (0.3, 3.1)
                (0.3, 0) node [right, align=left] {$O(k \nu)$ \\ tests}
            ;
        \end{scope}
    \end{tikzpicture}
    \caption{
        Bit-mixing coding \cite{BCS21}: Each person chooses a random
        mask and inject her phone number through the holes.  The decoder
        has to enumerate over all masks so the trick here is to
        predefine a set of $\approx k^2$ masks to (1) limit decoding
        complexity while (2) avoiding two sick people choosing the same
        mask.
    }                                               \label{fig:vertical}
\end{figure*}
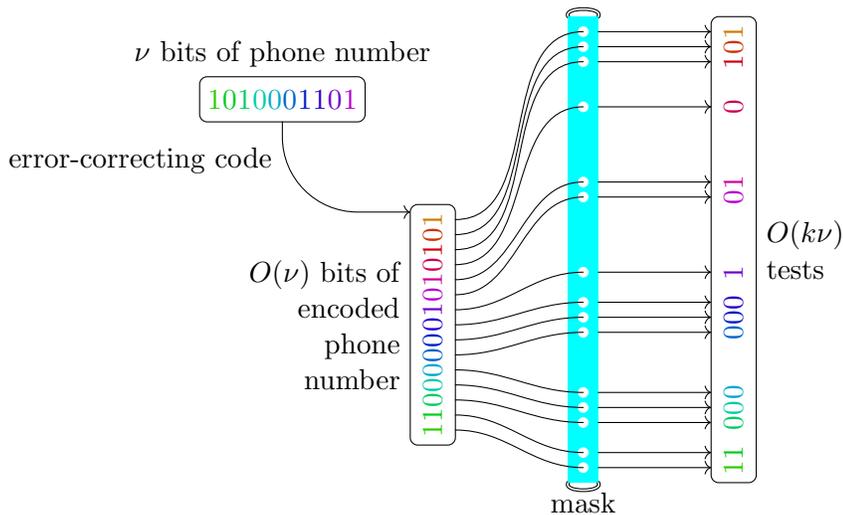

\subsection{Bit-mixing coding}

    Bit-mixing coding \cite{BCS21} is an involved
    copy-and-selectively-paste technique.  The encoder first prepares a
    set of $s \approx k^2$ masks, each having $\OO(\nu)$ holes out of
    $\OO(k \nu)$ possible positions, as represented by the cyan belt in
    Figure~\ref{fig:vertical}.  Each person encodes her phone number by
    an error-correcting code whose dimension is $k$ and length is the
    number of holes.  She then chooses a random mask and copy-and-paste
    the encoded phone number to tests behind the holes.

    To decode, extract encoded phone numbers from locations where a mask
    has holes (expect errors, but that is the purpose of
    error-correcting code).  Unlike COMP, whose decoder has to enumerate
    over all $n$ people, the decoder of Price, Scarlett, and Tan
    \cite{BCS21} only needs to enumerate over all $s$ masks.  When $n >
    k^2$, this decoder's complexity $s \poly(\nu)$ outperforms COMP's
    decoding complexity $n \poly(\nu)$.

    That said, there is a fundamental constraint that, any two sick
    person that share the same mask cannot be decoded successfully.  By
    the birthday paradox, the collision probability is about $k^2/s$,
    hence the choice of parameter $s \approx k^2$.  But if we want to
    lower the collision probability to, say, $1/k$, this forces $s
    \approx k^3$.  Such a trade-off might not always be favorable.  One
    of the innovations of this paper is to avoid enumerating masks (or
    over any other list longer than $k^{1+\varepsilon}$) but explicitly
    inserting long strings of hints, which we will call birthdays.  We
    can insert birthdays because there will be more than one bits
    passing each hole of the mask, as is seen by comparing
    Figure~\ref{fig:vertical} to Figure~\ref{fig:QR}.

\subsection{List-recoverable codes}

    List-recoverable codes are a general framework that, given a pile of
    puzzle pieces coming from several pictures, identifies puzzle pieces
    of the same origin as well as assembling them together.  More
    formally, it is a codebook $\CC \subset \Sigma^n$ such that, given
    subsets of available letters at each location $S_1, S_2, \dotsc, S_n
    \subset \Sigma$, one can easily find codewords that consist of
    available letters, i.e., codewords in the intersection
    \[ \CC \cap (S_1 \times S_2 \times \dotsb \times S_n). \]
    List-recoverable codes are very popular
    \cite{INR10, NPR11, NPR12, GNP13, GLP14, ARF22, DoW22}
    among sparse recovery works for the obvious reason that the entirety
    of this type of problems is to isolate signals from a mixture.

    In this work, our construction is equivalent to a more general
    notion of list-recoverable codes: Given $S_1, \dotsc, S_n \subset
    \Sigma$, we want codewords that use available letters at a few (not
    all) locations, and can use arbitrary letters elsewhere.  In
    literature, the coordinates that does not lie in $S_i$ are called
    errors.  We discover that list-recoverable codes allowing errors are
    the ``correct math object'' that we can run mathematical induction
    on (Figure~\ref{fig:nest}).  We thereby obtain a GT scheme where the
    number of layers of list-recoverable codes translates to the
    exponent of the exponent of the number of FNs and FPs
    (Figures \ref{fig:pyramid} and \ref{fig:evolve}).

\begin{figure*}
    \pgfmathdeclarefunction{g}{2}{%
        \pgfmathsetmacro\a{mod(#2*88, 207)}%
        \pgfmathsetmacro\b{mod(#2*\a, 309)}%
        \pgfmathsetmacro\c{mod(#2*\b, 411)}%
        \pgfmathsetmacro\d{mod(#2*\c, 513)}%
        \pgfmathparse{
            1 + sin(50*#1-\d)/2 + sin(110*#1-\c)/3 + sin(190*#1-\b)/4
        }%
    }
    \begin{tikzpicture} [baseline=0]
        \draw
            (2.5, 3) node [scale=0.8] {$\Bigl[
                \biggl\{ \substack{f(p_1) \\ g(p_1) \\ h(p_1)} \biggr\},
                \biggl\{ \substack{f(p_2) \\ g(p_2) \\ h(p_2)} \biggr\},
                \biggl\{ \substack{f(p_3) \\ g(p_3) \\ h(p_3)} \biggr\},
                \biggl\{ \substack{f(p_4) \\ g(p_4) \\ h(p_4)} \biggr\}
            \Bigr]$}
            [domain=0:5, samples=200]
            plot (\x, {g(\x, 1)}) node [right] {$f$}
            plot (\x, {g(\x, 2)}) node [right] {$g$}
            plot (\x, {g(\x, 3)}) node [right] {$h$}
            (0, 0) -- (5, 0)
            foreach \x in {1, ..., 4}{
                (\x, 0) circle (1pt) node [below] {$p_\x$}
                (\x, {g(\x, 1)}) circle(1pt)
                (\x, {g(\x, 2)}) circle(1pt)
                (\x, {g(\x, 3)}) circle(1pt)
                (\x, {min(g(\x, 1), g(\x, 2), g(\x, 3))})
                edge [->, dotted] (\x, 2.5)
            }
        ;
    \end{tikzpicture}
    \hfill
    \begin{tikzpicture} [baseline=0]
        \draw [domain=0:5, samples=200]
            plot (\x, {g(\x, 1)})
            plot (\x, {g(\x, 2)})
            plot (\x, {g(\x, 3)})
            (0, 0) -- (5, 0)
        ;
        \foreach \X in {1, ..., 4} {
            \draw
                (\X, 0) circle (1pt) node [below] {$p_\X$}
                (\X, {g(\X, 1)}) circle(1pt)
                (\X, {g(\X, 2)}) circle(1pt)
                (\X, {g(\X, 3)}) circle(1pt)
                (\X, {min(g(\X, 1), g(\X, 2), g(\X, 3))})
                edge [->, dotted] (\X, 2.25)
            ;
            \draw [shift={(\X, 2.5)}, scale=1/6]
                [domain=-2:2, samples=50]
                plot (\x, {3 * g(\x, 10+\X)})
                plot (\x, {3 * g(\x, 20+\X)})
                plot (\x, {3 * g(\x, 30+\X)})
                (-2, 0) -- (2, 0)
            ;
        }
    \end{tikzpicture}
    \hfill
    \begin{tikzpicture} [baseline=0]
        \draw [domain=0:5, samples=200]
            plot (\x, {g(\x, 1)})
            plot (\x, {g(\x, 2)})
            plot (\x, {g(\x, 3)})
            plot (\x, {g(\x, 4)})
            plot (\x, {g(\x, 5)})
            plot (\x, {g(\x, 6)})
            (0, 0) -- (5, 0)
        ;
        \foreach [count=\X] \a/\b/\c in {1/2/3, 1/4/5, 2/4/6, 3/5/6} {
            \draw
                (\X, 0) circle (1pt) node [below] {$p_\X$}
                (\X, {g(\X, \a)}) circle(1pt)
                (\X, {g(\X, \b)}) circle(1pt)
                (\X, {g(\X, \c)}) circle(1pt)
                (\X, {min(g(\X, \a), g(\X, \b), g(\X, \c))})
                edge [->, dotted] (\X, 2.25)
            ;
            \draw [shift={(\X, 2.5)}, scale=1/6]
                [domain=-2:2, samples=50]
                plot (\x, {3 * g(\x, 10+\X)})
                plot (\x, {3 * g(\x, 20+\X)})
                plot (\x, {3 * g(\x, 30+\X)})
                (-2, 0) -- (2, 0)
            ;
        }
    \end{tikzpicture}

    \caption{
        Left: list-recoverable code.  Middle: nested list-recoverable
        codes.  Right: nested list-recoverable codes that allow errors.
    }                                                   \label{fig:nest}
\end{figure*}
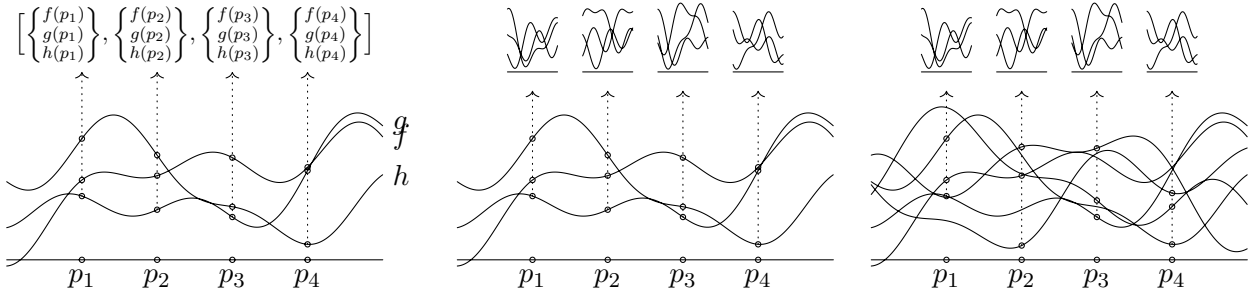

\section{Old Group Testing Results}                  \label{app:results}

    In his appendix we classify old GT results into three tables
    for ease of comparison.
    
\paragraph{Table~\ref{tab:k^2}:
           nonadaptive combinatorial group testing.}
    
    This table is special in that $\Omega(k^2 \log_k n)$ tests are
    needed.  If we allow either adaptive tests or FPs or FNs, the number
    of tests suddenly drops to $k \poly(\nu)$.  The best results we know
    is Cheraghchi--Nakos's Theorems 13 and 14 \cite{ChN20}.  If tests
    are noisy, we speak of the \emph{worst-case} failure probability
    where $x$ is adversarial and $Z$ is independently random.
    Interesting is that, in this case, the lower bound on the number of
    tests is also $\Omega(k^2 \log_k n)$.  Our contribution here is the
    observation that the term $k^2$, on a hand-waiving level, can be
    interpreted as repeating a partial-recovery GT about $k$ times to
    reduce the failure probability to the order of $1/\binom nk$.

\paragraph{Table~\ref{tab:fast}:
           nonadaptive tests and low-complexity decoder.}

    There are two cases.  If tests are noiseless, then Price--Scarlett
    \cite{PrS20} and Cheraghchi--Nakos \cite{ChN20} is the best GT
    scheme (by the analysis in \cite{BonsaiGT}).  If tests are noisy, we
    either have to pay extra tests (GROTESQUE \cite{CJB17}, a factor of
    $\log k$) or extra complexity (bit-mixing coding \cite{BCS21}, a
    factor of $k$).  Recently, Price--Scarlett--Tan \cite{PST23}
    provides a continuous trade-off between GROTESQUE and bit-mixing
    coding.  Our contribution here is to provide a new parametric
    trade-off that is incomparable with Price--Scarlett--Tan.

\paragraph{Table~\ref{tab:precise}:
           achieving the information-theoretical lower bound on $m$
           (or off by an explicit constant factor).}
    
    The set of sick people needs $\log_2 \binom nk$ bits to describe so
    the ultimate goal is to make $m \sim C(Z)^{-1} \log_2 \binom nk$,
    where $C(Z)$ is the capacity of $Z$.  The best GT schemes here are
    $2$-round SPIV \cite{CGH21} and those by Scarlett \cite{Sca19n,
    Sca19a}, also requiring adaptive tests.  We, however, prefer
    nonadaptive tests and aim for \emph{order-optimal} $m = \OZ(k\nu)$,
    where $\OZ$ is a big-O constant depending on $Z$.

\begin{table*}
    \caption{
        Nonadaptive combinatorial GT; $(\nu, \kappa) \coloneqq
        (\log_2 n, \log_2 k)$.  Tests are noisy if $\OO$ is
        replaced by $\OZ$ (meaning a constant depending on $Z$).
    }                                                    \label{tab:k^2}
    \bigskip
    \centering
    \leftskip-9cm plus9cm
    \rightskip-9cm plus9cm
    \begin{tabular}{cccc}
        \toprule

        Name and reference & $m$ (\#tests) & $\DD$'s complexity & Remark
        
        \\ \midrule

        Kautz--Singleton \cite{KaS64}
        & $\OO(k^2 \nu^2 / \kappa^2)$ & unspecified & explicit $A$
        
        \\ Indyk--Ngo--Rudra \cite[Cor~1]{INR10}
        & $\OO(k^2 \nu)$ & $\poly(m)$ &

        \\ Ngo--Porat--Rudra \cite[Thm~16]{NPR11}
        & $\OO(k^2 \nu)$ & $\poly(m)$ & explicit $A$
        
        \\ Porat--Rothschild \cite[Thm~1]{PoR11}
        & $\OO(k^2 \nu)$ & $\OO(mn)$ & explicit $A$
        
        \\ Cheraghchi--Ribeiro \cite[Thm~19]{ChR23}
        & $\OO(k^2 \nu^2)$ & $k^3 \poly(\nu)$ & explicit $A$

        \\ Cheraghchi--Nakos \cite[Thm~13]{ChN20}
        & $\OO(k^2 \min(\nu, \nu^2 / \kappa^2))$
        & $\OO(m + k (\nu - \kappa)^2)$ &

        \\ Cheraghchi--Nakos \cite[Thm~14]{ChN20}
        & $\OO(k^2 \nu)$ & $m \poly(\nu)$ & explicit $A$

        \\ Atia--Saligrama \cite[Rem~VI.1]{AtS12}
        & $\OZ(k^2 \nu)$ & unspecified & FP channel; worst-case

        \\ Atia--Saligrama \cite[Rem~VI.2]{AtS12}
        & $\OZ(k^2 \nu \kappa^2)$ & unspecified & FN channel; worst-case

        \\ \cmidrule(l{3em}r{3em}){1-4}

        \textbf{Worst-case Gacha}
        $(\tau, \sigma) \leftarrow (\sqrt{\log_2 \nu}, k2^\tau)$
        & $\OZ\bigl( k^2 \nu 2^{\OO(\sqrt{\log\nu})} \bigr)$
        & $\OZ(k^2 \poly(\nu))$ & any channel; worst-case

        \\ \bottomrule
    \end{tabular}
\end{table*}

\begin{table*}
    \caption{
        GTs that aim to minimize $m$ and $\DD$'s complexity;
        $(\nu, \kappa) \coloneqq (\log_2 n, \log_2 k)$. 
        Tests are noisy if $\OO$ is replaced by $\OZ$
        (meaning a constant depending on $Z$).
    }                                                   \label{tab:fast}
    \bigskip
    \centering
    \leftskip-9cm plus9cm
    \rightskip-9cm plus9cm
    \begin{tabular}{cccc}
        \toprule

        Name and reference & $m$ (\#tests) & $\DD$'s complexity & Remark

        \\ \midrule

        Cheraghchi--Nakos \cite[Thm~20]{ChN20}
        & $\OO(k \nu)$ & $\OO(k \nu)$ &

        \\ Price--Scarlett \cite[Thm~1]{PrS20}
        & $\OO(k \nu)$ & $\OO(k \nu)$ &

        \\ Nonadaptive splitting \cite[Thm~1]{BonsaiGT}
        & $\frac{1 + \varepsilon}{\ln2} k \nu$
        & $\OO(\varepsilon^{-2} k \nu)$ &

        \\ Nonadaptive splitting \cite[Thm~2]{BonsaiGT}
        & $\frac{1 + \varepsilon}{\ln2} k \max(\nu - \kappa, \kappa)$
        & $\OO(\varepsilon^{-2} k^2 \nu)$ &

        \\ Atia--Saligrama \cite[Thm~V.2]{AtS12}
        & $\OO(k \nu \kappa^2)$ & unspecified &

        \\ Atia--Saligrama \cite[Thm~VI.2]{AtS12}
        & $\OZ(k \nu \kappa^2)$ & unspecified & FP channel 

        \\ Atia--Saligrama \cite[Thm~VI.5]{AtS12}
        & $\OZ(k \nu \kappa^2)$ & unspecified & FN channel

        \\ Noisy COMP \cite[Thm~15]{CJS14}
        & $\OZ(k \nu)$ & $\OZ(m n)$ &

        \\ Noisy DD \cite[Thm~5]{ScJ20}
        & $\OZ(k \nu)$ & $\OZ(mn)$ &

        \\ Mazumdar \cite[Thm~1]{Maz16}
        & $\OZ(k \nu^2 / \kappa)$ & $\OZ(mn)$ & explicit $A$

        \\ IKWÖ via KS \cite[Thm~2]{IKW19}
        & $\OZ(k \nu)$ & $\OZ(mn)$ & explicit $A$

        \\ IKWÖ's fast KS \cite[Thm~3]{IKW19}
        & $\OZ(k \nu \log(\nu/\kappa))$
        & $\OZ(k^3 \nu \log(\nu/\kappa))$ & explicit $A$

        \\ Inan--Özgür's fast KS \cite[Thm~2]{InO20}
        & $\OO(k \nu)$ & $\OO(k^3\kappa + k \nu)$ & explicit $A$

        \\ Bit-mixing coding \cite[Sec~V]{BCS21}
        & $\OZ(k \nu)$ & $\OZ(k^2 \kappa \nu)$ &
        
        \\ Price--Scarlett--Tan \cite[Thm~4.1]{PST23}
        & $\OZ(t k \nu)$ & $\OZ((k (\nu - \kappa))^{1+\varepsilon})$
        & $k^{1-t\varepsilon}$ FPs \& FNs

        \\ SAFFRON \cite[Thm~7]{LCP19}
        & $\OZ(k \nu)$ & $\OZ(k \nu)$ & $\varepsilon k$ FPs \& FNs

        \\ Bi-regular SAFFRON \cite[Thm~14]{VJN17}
        & $\OZ(k (\nu - \kappa))$
        & $\OZ(k (\nu - \kappa))$ & $\varepsilon k$ FPs \& FNs

        \\ Partial-recovery GROTESQUE \cite[Cor~8]{CJB17}
        & $\OZ(k \nu)$ & $\OZ(k \nu)$ & $\varepsilon k$ FPs \& FNs

        \\ Exact-recovery GROTESQUE \cite[Thm~2]{CJB17}
        & $\OZ(k \nu \kappa)$ & $\OZ(k(\nu + \kappa^2))$ &
        
        \\ $2$-round GROTESQUE \cite[Thm~3]{CJB17}
        & $\OZ(k(\nu + \kappa^2))$ & $\OZ(k(\nu + \kappa^2))$ &

        \\ $\OO(\kappa)$-round GROTESQUE \cite[Thm~1]{CJB17}
        & $\OZ(k \nu)$ & $\OZ(k \nu)$ &

        \\ \cmidrule(l{3em}r{3em}){1-4}

        \textbf{Partial-Recovery Gacha}
        $(\tau, \sigma) \leftarrow (2, 1)$
        & $\OZ(k \nu)$ & $\OZ(k \poly(\nu))$
        & $ k \exp(-\nu^{1-1/\tau})$ FPs \& FNs

        \\ \textbf{Exact-recovery Gacha}
        $(\tau, \sigma) \leftarrow (2, 1 + \kappa / \sqrt\nu)$
        & $\OZ(k (\nu + \kappa \sqrt\nu))$ & $\OZ(k \poly(\nu))$
        & $o(1)$ FNs \& FPs 

        \\ \textbf{Order-optimal Gacha}
        $(\tau, \sigma) \leftarrow (\OO(1), 1)$
        & $\OZ(k \nu)$ & $\OZ(k \poly(\nu))$
        & $o(1)$ FNs \& FPs if $\kappa < \nu^{1-1/\OO(1)}$

        \\ \bottomrule
    \end{tabular}
\end{table*}

\begin{table*}
    \caption{
        Works that try to achieve capacity or some explicit code rate;
        $(\nu, \kappa) = (\log_2 n, \log_2 k)$.  GT is nonadaptive if
        \#rounds is $1$.  The explicit multipliers $2.09 \geq \bar a
        \geq \check b \geq \breve c \geq \grave d$ are plotted in
        Figure~\ref{fig:precise}.  All except the last two ``$m$''
        entries are up to a factor of $1 + o(1)$.
    }                                                \label{tab:precise}
    \bigskip
    \centering
    \leftskip-9cm plus9cm
    \rightskip-9cm plus9cm
    \begin{tabular}{ccccc}
        \toprule

        Name and reference
        & \#rounds & $m$ (\#tests) & $\DD$'s complexity & Remark

        \\ \midrule
        
        COMP \cite[Thm~2]{JAS19}
        & $1$ & $\bar a k \nu$ & $\OO(mn)$ &

        \\ DD \cite[Thm~3]{JAS19}
        & $1$ & $\check b k \nu$ & $\OO(mn)$ &

        \\ Nonadaptive splitting \cite[Thm~1]{BonsaiGT}
        & $1$ & $(\bar a + \varepsilon) k \nu$
        & $\OO(\varepsilon^{-2} k \nu)$ &
        
        \\ Nonadaptive splitting \cite[Thm~2]{BonsaiGT}
        & $1$ & $(\check b + \varepsilon) k \nu$
        & $\OO(\varepsilon^{-2} k^2 \nu)$ &
        
        \\ SSS \cite[Thm~I.1]{CGH20}
        & $1$ & $\breve c k \nu$ & unspecified &

        \\ SPIV \cite[Thm~1.2]{CGH21}
        & $1$ & $\breve c k \nu$
        & $\poly(n)$ & $k < n^{1-\Omega(1)}$

        \\ Dense SPIV \cite[Cor~2.1]{BSP22}
        & $1$ & $\breve c k \nu$
        & $\poly(n)$ & all $k < n$

        \\ $2$-round SPIV \cite[Thm~1.3]{CGH21}
        & $2$ & $\grave d k \nu$ & $\poly(n)$ &

        \\ Damaschke--Muhammad \cite[Thm~3]{DaM12}
        & $2$ & $k \nu + {}$reminder & unspecified & verify

        \\ Damaschke--Muhammad \cite[Thm~2]{DaM12}
        & $3$ & $\grave d k \nu + {}$reminder & unspecified &

        \\ Scarlett's $2$-round \cite[Thm~1]{Sca19n}
        & $2$ & $\grave d k \nu / C(Z) + \OZ(k\kappa)$
        & unspecified & BSC
        
        \\ Scarlett's $3$-round \cite[Thm~3]{Sca19n}
        & $3$ & $\grave d k \nu / C(Z) + \OZ(k\kappa)$
        & unspecified & BSC
        
        \\ Scarlett's $4$-round \cite[Thm~1]{Sca19a}
        & $4$ & $\grave d k \nu / C(Z) + \OZ(k\kappa)$
        & $\OZ(\varepsilon k^{2+\varepsilon})$ & BSC
        
        \\ Hwang's splitting \cite[Thm~1]{Hwa72}
        & $m$ & $\grave d k \nu + \OO(k)$ & unspecified & zero-error
        
        \\ Allemann's splitting \cite[Thm~2]{All13}
        & $m$ & $\grave d k \nu + k/5 + \dotsb$
        & unspecified & zero-error

        \\ \bottomrule
    \end{tabular}
\end{table*}

\begin{table*}
    \caption{
        References for achievability and impossibility bounds of some
        named GT strategies plotted in Figure~\ref{fig:precise}.  B
        stands for Bernoulli design; C stands for near-constant column
        weight design.  COMP stands for combinatorial orthogonal
        matching pursuit.  DD stands for definite defective.  SCOMP
        stands for sequential COMP.  LP stands for linear programming.
        SSS stands for smallest satisfying set, the peak of nonadaptive
        (one-round) GT.  For SSS, see also \cite[Corollary~2.1]{BSP22}
        and \cite[Theorem~4]{FlM21}.
    }                                                   \label{tab:both}
    \bigskip
    \centering
    \begin{tabular}{ccc}
        \toprule

        Design-decoder & achievability bound & impossibility bound

        \\ \midrule
        
        B-COMP
        & \cite[Theorem~4]{CCJ11}
        & \cite[Lemma~2]{Ald17}

        \\ C-COMP
        & \cite[Theorem~2]{JAS19}
        & \cite[Theorem~2]{JAS19}
        
        \\ B-DD
        & \cite[Theorem~12]{ABJ14}
        & ---
        
        \\ C-DD
        & \cite[Theorem~3]{JAS19}
        & \cite[Theorem~I.2]{CGH20}
        
        \\ B-SCOMP ($\leq$ B-DD)
        & \cite[Theorem~8]{Ald17a}
        & ---

        \\ C-SCOMP (= C-DD)
        & \cite[Theorem~8]{Ald17a}
        & \cite[Theorem~I.2]{CGH20}

        \\ B-LP ($\leq$ B-DD)
        & \cite[Theorem~9]{Ald17a}
        & ---

        \\ C-LP ($\leq$ C-DD)
        & \cite[Theorem~9]{Ald17a}
        & ---

        \\ B-SSS
        & \cite[Corollary~6]{ScC17}
        & \cite[Theorem~1]{Ald17}

        \\ C-SSS
        & \cite[Theorem~I.1]{CGH20}
        & \cite[Theorem~4]{JAS19}

        \\ \bottomrule
    \end{tabular}
\end{table*}

\begin{figure*}
    \centering
    \begin{tikzpicture} [x=14cm, y=4cm]
        \footnotesize
        \colorlet{COMP}{Orange!80!black}
        \colorlet{DD}{Plum!80!black}
        \colorlet{SSS}{Emerald!80!black}
        \colorlet{INFO}{Salmon!80!black}
        \draw [->] (0, 0) -- (0, 2)
            node [above] {$\dfrac{m}{k\log_2 n}$};
        \draw [->] (0, 0) -- (1, 0) -- +(1em, 0)
            node [right] {$\log_n k$};
        \path
            (0, {e*ln(2)}) coordinate (NW)
            (1, {e*ln(2)}) coordinate (NE)
            (0, {1/ln(2)}) coordinate (WNW)
            (1, {1/ln(2)}) coordinate (ENE)
            (0, 1) coordinate (W)
            (0, 0) coordinate (SW)
            (1, 0) coordinate (SE)
            (1/2, {e*ln(2)/2}) coordinate (M1)
            (1/2, {1/(2*ln(2))}) coordinate (M2)
            ({1/(1+e*ln(2))}, {e*ln(2)/(1+e*ln(2))}) coordinate (M3)
            ({ln(2)/(1+ln(2))}, {1/(1+ln(2))}) coordinate (M4)
            (NW) node [left] {$e \ln2$}
            (WNW) node [left, inner sep=1pt] {${1}/{\ln2}$}
            (W) node [left]  {$1$}
            (SW) node [left] {$0$} node [below] {$0$}
            (SE) node [below] {$1$}
        ;
        \draw [dotted]
            (NW) -- (SE) (SW) -- (NE)
            (WNW) --  (SE) (SW) -- (ENE)
        ;
        \draw [COMP]
            (NW) -- node [above] {B-COMP} (NE)
            (WNW) -- node [above] {C-COMP $= \bar a$} (ENE)
        ;
        \draw [DD]
            (NW) --
            node [auto, sloped] {B-DD}
            node [pos=1, right, sloped] {$(\frac12, \frac{e\ln2}{2})$}
            (M1) --
            node [auto, sloped] {B-DD}
            (NE)
            (WNW) --
            node [auto, sloped] {C-DD $= \check b$}
            node [pos=1, right, sloped] {$(\frac12, \frac{1}{2\ln2})$}
            (M2) --
            node [auto, sloped] {C-DD $= \check b$}
            (ENE)
        ;
        \draw [SSS, transform canvas={shift={(1pt, 0)}}]
            (W) --
            node [auto, sloped] {B-SSS $=$ C-SSS $= \breve c$}
            (M3) --
            node [auto, ', sloped] {B-SSS}
            node [pos=0, left, sloped, rotate=-0]
            {$(\frac{1}{1+e\ln2}, \frac{e\ln2}{1+e\ln2})$}
            (NE)
            (M3) --
            (M4) --
            node [pos=0, left, sloped, rotate=0]
            {$(\frac{\ln2}{1+\ln2}, \frac{1}{1+\ln2})$}
            node [auto, ', sloped] {C-SSS $= \breve c$}
            (ENE)
        ;
        \draw [INFO]
            (W) --
            node [pos=.6, auto, ', sloped] {counting bound $= \grave d$}
            (SE)
        ;
    \end{tikzpicture}
    \caption{
        Explicit code rates of some GT strategies listed in
        Table~\ref{tab:precise}.  See Table~\ref{tab:both} for
        references to both direct and converse bounds.  Similar plots
        are found in \cite[Fig.~1]{BonsaiGT} and \cite[Figure~1]{CGH21}.
        Up-side-down plots are found in \cite[Fig.~2]{CGH20},
        \cite[Figure~2.1]{AJS19}, \cite[Fig.~1]{JAS19},
        \cite[Figure~3]{ScC17}, \cite[Fig.~1]{Ald17},
        \cite[Fig.~2]{Ald17a}, \cite[Figure~1]{ScC16}, and
        \cite[Fig.~2]{ABJ14}.
    }                                                \label{fig:precise}
\end{figure*}
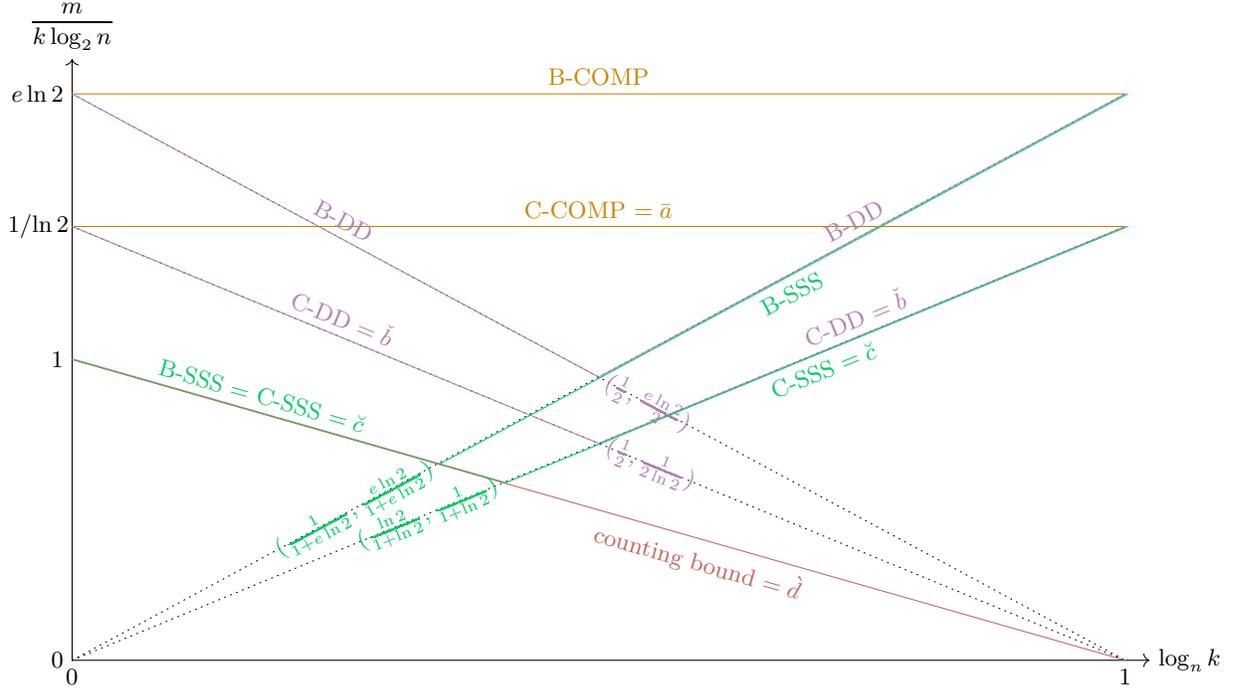

\begin{table}
    \centering
	\caption{
        Variants of GT inspired by applications.
	}                                                  \label{tab:quant}
    \bigskip
	\begin{tabular}{ccc}
		\toprule

        Regime & Readings & Remixing

        \\ \midrule

        Binary GT
        & negative, positive
        & negative $\vee$ positive $=$ positive

		\\ Quantitative GT \cite{GHK22}
        & $0, 1, 2, 3, 4, 5, 6, 7, \dotsc$
        & $8 + 9 = 17$

		\\ Semi-quantitative GT \cite{CGM21}
        & $[0,3), [3,6), [6,9), \dotsc$
        & $[0,3) + [3,6) = [3,9)$   
        
		\\ Tropical GT \cite{TropicalGT23}
        & $2^{-1}, 2^{-2}, \dotsc, 2^{-\infty}$
        & $2^{-30} + 2^{-15} \coloneqq 2^{-15}$

        \\ Threshold GT \cite{BKC19}
        & $(< \theta)$, $(> \theta)$
        & usual addition

        \\ Generalized GT \cite{CJZ23} 
        & negative, positive
        & $a + \dotsb + z \leadsto$ Bern$(1 - \frac{1}{a+\dotsb+z})$

        \\ One-bit compressed sensing \cite{MaM24}
        & $(< 0)$, $(> 0)$
        & $5 + (-2) \leadsto (> 0)$

		\\ \bottomrule
	\end{tabular}
\end{table}

\begin{table*}
    \caption{
        Many problems are related to GT.
        Other indirect applications include pattern matching
        \cite{CEP10} and dead sensor discovery \cite{GoH08}.
    }                                                 \label{tab:models}
    \bigskip
    \centering
    \begin{tabular}{cccc}
        \toprule

        Problem
        & There are $n$ & $k$ of them & To solve, use $m$

        \\ \midrule

        Disease control \cite{AlE22}
        & people & are sick & virus tests

        \\ Genotyping \cite{EGB10}
        & genes & cause cancer & gene tests

        \\ Wireless channel reservation \cite{LuG08}
        & cellphones & want to talk & frequencies

        \\ Wireless direct transmission \cite{LFP22}
        & messages & are sent & frequencies

        \\ Heavy hitter \cite{ChN20}
        & items & are popular  & words of storage
        
        \\ Compressed sensing \cite{ChN20}
        & signals & are active  & measurements

        \\ DoS attack \cite{KGM08, YIT10}
        & users & are spamming  & virtual servers

        \\ Traitor tracing \cite{CFN00}
        & users & resell keys & keys

        \\ Computer forensics \cite{GAT05}
        & files & will be modified & bits of storage

        \\ Property-preserving hashing \cite{Min22}
        & properties & appears in a file & bits per file

        \\ Image compression \cite{HoL00}
        & wavelet coefficients & are nonzero & bits per digit

        \\ \bottomrule
    \end{tabular}
\end{table*}

\section{Downstream Implications in Heavy Hitter, Compressed Sensing,
         and Multiple Access Channel}                  \label{app:imply}

    In this appendix, we provide a simplified narrative that
    demonstrates how could a good GT design influences other related
    problems, using heavy hitter, compressed sensing, and multiple
    access channel as examples.  See Table~\ref{tab:models} for more
    possible connections and see Table~\ref{tab:quant} for how the
    features of the target problems back-propagate to GT.

    It is worth noting that quantifying these implications is a very
    difficult task as these problems have become super competitive in
    recent years and we should simply double-check the parameters on a
    case-by-case basis.

\paragraph{Group testing $\leadsto$ heavy hitter.}

    The heavy hitter problem is about finding the top $k$ popular, say,
    videos in the past hour so we can cache the video files in more
    servers.  A one-liner solution is to borrow a configuration matrix
    $A \in \{0, 1\}^{m\times n}$ from a GT scheme and compute the sum of
    columns $\yy \coloneqq \sum_j A_{v_j}$, where $v_j$ is the index of
    the $j$th requested video and $A_v$ is the $v$th column of $A$.  It
    is easy to see that $\yy = A\xx$, where $\xx$ is the popularity
    vector of the videos.  We then threshold the result as $y_i
    \coloneqq (\yy_i > \theta)$ for all $i \in [m]$.  Seeing $y$, the GT
    decoder will tell us which videos are popular.  If we also want to
    know how popular they are, the standard practice is to combine GT
    with the count sketch \cite{CCF04} or the count--min sketch
    \cite{CoM05a}.

\paragraph{Heavy hitter $\leadsto$ compressed sensing.}

    Now that we have a recipe for constructing heavy hitter sketchers,
    we can build a compressed sensing solver that solves $\yyy \coloneqq
    A \xxx + \text{errors}$, where $\xxx$ might contain negative or even
    imaginary numbers, as follows.  We can take the real part of $\yyy$
    or the modulus part of $\yyy$, feed that to a heavy hitter sketcher,
    and obtain a sketch $\xxx'$ of $\xxx$.  Compute $\yyy' \leftarrow
    \yyy - A\xxx'$, we now solve $\yyy' = A (\xxx - \xxx') +
    \text{errors}$ for $(\xxx - \xxx')$, which will give us the
    secondary sketch $\xxx''$ that approximates the residual.  We can
    repeat this process until $\xxx' + \xxx'' + \dotsb$ approximates
    $\xxx$ well enough.

\paragraph{Compressed sensing $\leadsto$ multiple access channel.}

    Now that we have a compressed sensing solver, we remark that this
    supply chain extends to multiple access channels as wireless devices
    who have messages to send can transmit a column of $A$ and the base
    station will receive the column sum $\sum_j e^{s+pi} A_{d_j}$ plus
    noise, where $d_j$ is the $j$th active device, $s$ is the strength
    of the signal, and $p$ is the phase (due to imperfect clock).  What
    attracts interest in this supply chain is that the parameters of the
    final multiple access solution depends heavily on that of compressed
    sensing, which is largely determined by that of heavy hitter, which
    originates from and is deeply influenced by that of GT.  Studying
    GT for better parameters and conceptually simpler schemes radiates
    impacts on the downstream applications and is what we are trying to
    do here.

\bibliographystyle{alphaurl}
\bibliography{GachaGT-21.bib}

\end{document}